\documentclass[prx,twocolumn,showpacs,amsfonts,amsmath,floatfix,superscriptaddress]{revtex4-2}
\usepackage{natbib}
\usepackage{amsthm, amsmath,amssymb}
\usepackage{graphicx}
\usepackage{dcolumn}
\usepackage{hyperref}
\usepackage{subcaption}
\usepackage{ragged2e}
\hypersetup{
    colorlinks=true,
    linkcolor=blue,
    filecolor=blue, 
    citecolor=blue,
    urlcolor=blue,
}
\usepackage{braket}
\usepackage{bm}

\graphicspath{%
  {figs/}%
  {rec_circs/}%
}

\usepackage{multirow}

\usepackage{booktabs}

\newcommand{\tr}{\mathrm{Tr}}

\newcommand{\cA}{\mathcal{A}}

\newcommand{\cC}{\mathcal{C}}

\newcommand{\cE}{\mathcal{E}}

\newcommand{\cO}{\mathcal{O}}
\newcommand{\cP}{\mathcal{P}}

\newcommand{\cR}{\mathcal{R}}
\newcommand{\cS}{\mathcal{S}}

\newcommand{\cW}{\mathcal{W}}

\newcommand{\m}{\mu}
\newcommand{\n}{\nu}
\newcommand{\Rl}{\rangle}
\newcommand{\Ll}{\langle}

\newtheorem{theorem}{Theorem}

\newtheorem{lemma}{Lemma}

\begin{document}

\title{Universal syndrome-based recovery for noise-adapted quantum error correction}

\author{Debjyoti Biswas}
\email{biswasd@chalmers.se}
\affiliation{Department of Microtechnology and Nanoscience (MC2),Chalmers University of Technology, SE-412 96 Göteborg, Sweden.}
\author{Prabha Mandayam}
\email{prabhamd@physics.iitm.ac.in}
\affiliation{Department of Physics}
\affiliation{Center for Quantum Information, Communication and Computing, \href{https://ror.org/03v0r5n49}{Indian Institute of Technology Madras}, Chennai 600036, India.}

\begin{abstract}
Quantum error correction (QEC) is an essential tool for quantum computing, enabling reliable information processing in the presence of noise. Syndrome measurements play a central role in QEC, enabling unambiguous identification of the location and type of errors. While syndrome extraction is natural for conventional QEC protocols, where the errors satisfy certain algebraic constraints \emph{perfectly}, this feature is largely missing in the framework of approximate or noise-adapted QEC. Rather, noise-adapted recovery maps, such as the Petz map, are used in the latter scenario, but implementing such tailored recovery processes on hardware can be quite challenging. Here,  we address this issue by proposing an algorithmic approach to identifying error syndromes for arbitrary codes and noise processes. We then use our algorithm to develop a variant of the Petz recovery map --- a syndrome-based Petz recovery map --- which can then be implemented via syndrome measurements. We demonstrate the efficacy of our approach in the context of amplitude-damping noise by constructing the syndrome-based Petz map for the $4$-qubit code. We execute our recovery circuits on IBM quantum hardware to demonstrate break-even performance of a noise-adapted QEC protocol, with up to threefold improvements in qubit $T_{1}$ times.
\end{abstract}

\maketitle

\section{Introduction}

Noise in quantum computers poses one of the most significant obstacles in achieving reliable and scalable quantum computation. Quantum error correction (QEC) \cite{bm_terhal}  paves the way for fault-tolerant quantum computation by encoding logical qubits into multiple physical qubits and using the properties of these logical states to detect and correct for errors. While standard QEC protocols are \emph{general-purpose}, in the sense that they do not require knowledge of the noise processes in the quantum hardware, having knowledge of the physical noise processes can often aid in designing resource-efficient QEC protocols. 
For example, the smallest quantum code requires five qubits to correct for arbitrary single-qubit errors~\cite{laflamme1996perfect}, whereas these exist \emph{noise-adapted} quantum codes that can correct for purely amplitude-damping (AD) noise using only four qubits \cite{leung}.

However, such noise-adapted QEC codes (see~\cite{jayashankar2023quantum}) for a recent review) may not always satisfy the well known the Knill-Lafllame (KL) conditions~\cite{KL-cond} exactly, but may only do so perturbatively. The standard bit-flip code and phase-flip three-qubit codes can be thought of as noise-adapted codes that satisfy the KL condition perfectly. On the other hand, The four-qubit Leung code \cite{leung}, which is tailored to protect against AD noise, satisfy a perturbed form of the KL conditions. Such codes that perturbatively satisfy the KL conditions belong to the category of approximate quantum error correcting codes. 

It is known that approximate quantum error correction (AQEC) schemes typically require lesser number of qubits \cite{fletcher2008channel} (or qudits \cite{dutta2025qudit}) to protect a single qubit (or qudits) against a known (or unknown) noise process. However, the framework brings with it other new challenges. One main challenge is that the noisy subspaces resulting from the action of the noise on the codespace  -- also referred to as the syndrome subspaces -- overlap and are non-orthogonal. Therefore, a generic AQEC protocol does not adopt the standard syndrome-based approach QEC; 
rather, alternate recovery schemes based on dynamical recovery maps have to be constructed, depending on the specific choice of code and noise process. 

Numerical approaches involving semi-definite programming \cite{fletcher2008channel} have been proposed to identify the optimal recovery and optimal codes \cite{liang_jiang_opt_code} in the AQEC framework. Recently, variational quantum algorithmic approaches have also been suggested \cite{grassl_vqe_codes, vqe_recs}, as also machine learning \cite{ML_code} based protocols to search for good AQEC codes.

On the analytical front, a universal recovery scheme based on the Petz map was shown to be near-optimal \cite{barnum2002, hkn_pm2010} for arbitrary codes and noise processes. The Petz map can also be the optimal decoder under certain conditions, as pointed out recently in~\cite{liang_jiang_optimality}. However, the known approaches for implementing the Petz map are not suitable for present day hardware since they require high depth circuits. In fact, the best circuit implementations for the Petz map \cite{gilyen2022_petz,biswas2024noise} have a gate count that scales exponentially with the length of the code. 


Figuring out simple decoding schemes or recovery operations has thus proved to be a significant bottleneck in implementing the AQEC framework. In our work, we address this issue by proposing a subspace orthogonalization algorithm that makes syndrome detection possible for the most general AQEC code. We show that the performance of the Petz map can be improved significantly through our subspace orthogonalisation algorithm, and for some cases, the performance of the Petz can be made optimal when the code and noise process do not satisfy the conditions in \cite{liang_jiang_optimality}.  


Next, we address the challenge of implementing the noise-adapted recovery on quantum computing hardware. We show that the subspace orthogonalization algorithm helps to design a hardware-efficient circuit for the Petz map. We then demonstrate the performance of the Petz map by performing a $T_1$ experiment with and without QEC and show a significant improvement in the $T_1$ lifetime. Notably, beyond break-even performance has thus far been demonstrated through the GKP code or qudit codes \cite{girvin_break_even_2022} on superconducting processors. Here, we report an improvement of the $T_1$ times on both IBM Heron and Eagle quantum processors, using an entirely noise-adapted QEC scheme for the first time.

The rest of the paper is structured as follows. We first provide an algorithm for syndrome orthogonalization for an arbitrary AQEC code and noise process in Sec.~\ref {sec:algorithm}. We then proceeds to design a Petz recovery map involving syndrome measurements in Sec.~\ref {sec:syn_petz}. Sec.~\ref {sec:ex} shows that this \emph{syndrome-based Petz map} serves as the optimal recovery for the $[[4,1]]$-Leung code. Finally, we propose a circuit implementation for the syndrome-based Petz map and demonstrate its performance on IBM quantum hardware in Sec.~\ref {sec:implementation}.

\section{Preliminaries: Approximate Quantum Error Correction and the Petz recovery map}

We begin by reviewing the framework of perfect and approximate quantum error correction (QEC). Consider a quantum noise channel $\cA$ with Kraus operators $\{A_k, k = 1,2,\ldots, N \}$. An $[n,k]$ quantum code $\cC$ is a $k$-dimensional subspace of the larger $n$-dimensional Hilbert space into which the information is encoded~\cite{chuang_nielsen}. The Knill-Laflamme (KL) conditions for perfect QEC states that a quantum code with projector $P$ corrects for errors $\{A_k\}$ if and only if there exist complex scalars $\lambda_{kl}$ such that, $PA_{k}^{\dagger}A_{l}P = \lambda_{kl}P$~\cite{KL-cond}. These conditions guarantee that (a) different errors map the codespace to mutually no-overlapping subspaces, and, (b) each error acts on the codespace in a unitary fashion. These two properties in turn guarantee the existence of a syndrome-based recovery that corrects perfectly for the errors $\{A_{i}\}$ on every state in the codespace. 

The code $\cC$ is said to be an approximate QEC (AQEC) code if it satisfies the Knill-Laflamme (KL) conditions approximately, that is~\cite{Bo},
\begin{align}\label{eq:beny_oreskov}
    P A_l^{\dagger}A_k P &= \lambda_{kl}P + P B_{kl}P. 
\end{align} The presence of the perturbation term $B_{kl}$ indicates that different errors do not map the code space to mutually orthogonal subspaces. Rather, the different error subspaces overlap with each to a certain degree depending on $||B_{kl}||$, where $||(.)||$ denotes the operator norm of the perturbation terms. Therefore, it is not possible to perform a syndrome measurement to detect and identify the errors unambiguously. 

For example, the $[4,1]$ code which is known to correct approximately for amplitude-damping noise~\cite{leung} gets mapped to a set of non-overlapping syndrome subspaces under the action of amplitude-damping noise. Although this code is stabilised by the group generated by $\cS \equiv\langle  XXXX, IIZZ, ZZII\rangle$, measuring the stabilizer generators in $\cS$ is inadequate to identify the location of the damping errors~\cite{fletcher2008channel}. One needs to measure two more operators $\{ ZIII, IIIZ\}$, which do not stabilise the codespace but rather stabilise the syndrome subspaces. After identifying the location of the damping errors using secondary syndromes, a noise-adapted recovery is performed using the procedure outlined in~\cite{leung}, or using the optimal recovery obtained through a semidefinite program (SDP)~\cite{fletcher2008channel}. 

An alternative approach to AQEC is based on identifying a completely positive trace-preserving (CPTP) map that can approximately reverse the action of the noise. An analytical construction of such a recovery channel, based on the Petz map, has been shown achieve near-optimal performance~\cite{barnum2002,hkn_pm2010} for arbitrary codes and noise channels. 
The Petz map corresponding to a noise channel $\cA$ acting on a code $\cC$ with projector $P$ is described by the following set of Kraus operators.
    \begin{equation}\label{eq:petz_def}
        \cR_{P,A} \sim \{  PA_i^{\dagger}\cA(P)^{-1/2} \}.
    \end{equation}
The Petz map is known to be near-optimal both in terms of the worst-case fidelity~\cite{hkn_pm2010} as well as the entanglement fidelity~\cite{barnum2002,liang_jiang_near_opt}. The worst-case fidelity achieved by a code $\cC$ and recovery $\cR$ for noise channel $\cA$ is given by, 
    \begin{align}
        F^2_{\rm min}(\cR \circ \cA )&= \underset{|\psi\rangle \in \cC}{\rm min} \,F^2(|\psi\rangle,\cR_{P,\cA} \circ \cA (|\psi\rangle\langle\psi|)),
    \end{align}
    where $F^2(|\psi\rangle,\cR \circ \cA (|\psi\rangle\langle\psi|))= \langle \psi| \cR \circ \cA (|\psi\rangle\langle\psi|) |\psi \rangle $ is the fidelity between a pure state in $\cC$ and the recovered state $(\cR \circ \cA )(|\psi\rangle\langle\psi|)$. The entanglement fidelity~\cite{Schumacher} is defined as, 
    \begin{align}
        F_{\rm Ent}(\cR \circ \cA)&= \frac{1}{d^2} \sum\limits_{k,l}\tr (|R_kA_l|^2),
    \end{align}
where $d$ is the dimension of the code-space and $\{R_k\}$s are the Kraus operators for the recovery map. 
    
More recently, it was shown that the Petz map achieves the optimal entanglement fidelity under certain conditions \cite{liang_jiang_optimality}. To state the optimality condition, we first define the QEC matrix $M:= \{M_{[mk],[n\ell]}= \langle m_L|A_{k}^{\dagger}A_\ell|n_L\rangle\}$ corresponding to a channel $\cA$ and a $d$-dimensional quantum codespace defined as the span of the codewords $\{|m_{L}\rangle, m \in [1,d]\}$. If $M$ satisfies  
    \begin{align}\label{eq:opt_qec_cond}
        [M,\tr_L (\sqrt{M}\otimes I_d)]=0,
    \end{align}
where $\tr_L (.)$ is the partial trace taken over the codespace and $I_{d}$ is the $d \times d$ identity matrix, then the Petz map was shown to achieve the maximum entanglement fidelity, for the given channel and codespace.
\begin{widetext}
\begin{center}
    
    \begin{figure}
\includegraphics[width=1.0\textwidth]{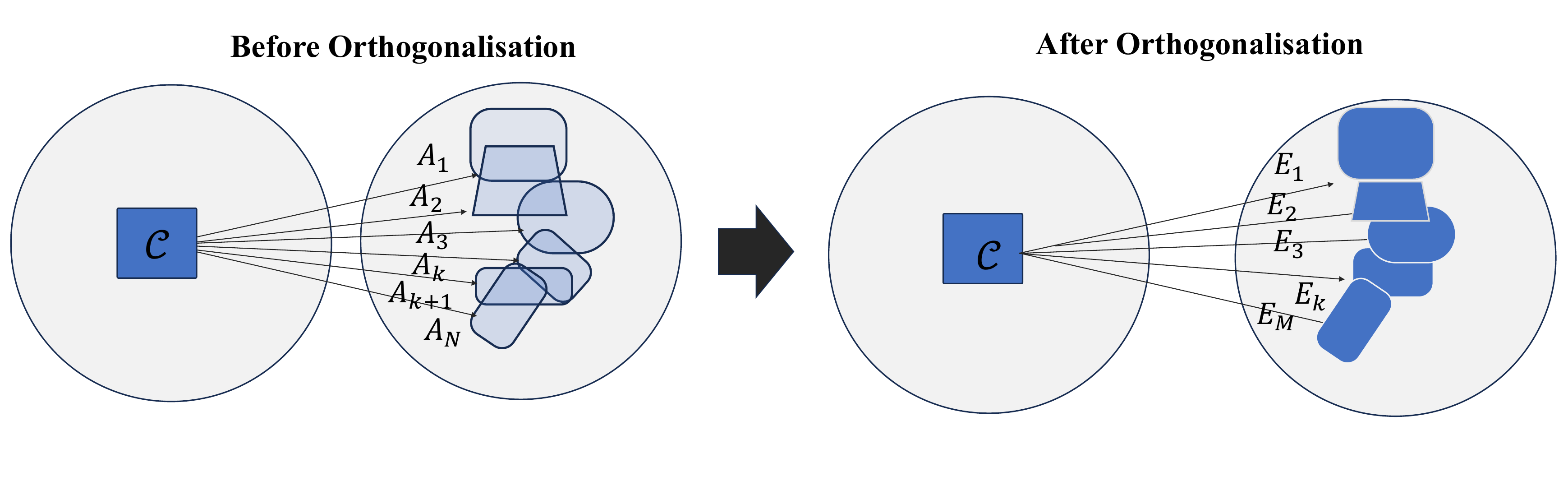}
\caption{The figure shows the action of the noise channel on the codespace before and after orthogonalization. The ${A_k}$ denote the Kraus operators of the noise channel, while the ${E_k}$ represent the Kraus operators obtained through the orthogonalisation process. The relationship between ${E_k}$ and ${A_k}$ is given in Eq.~\eqref{eq:Ek_forms}.}
        \label{fig:ortho_sub}
    \end{figure}  
    \end{center}
\end{widetext}

\section{Algorithm to generate orthogonal syndrome subspaces}

In this section, we present a solution to overcome the problem of overlapping error subspaces for arbitrary channels and codespaces. We develop an algorithmic procedure to obtain a noise channel {that is close to the original channel} and maps the codespace to a set of orthogonal subspaces.

Essentially, we seek to transform the set of noise operators $\{A_k\}$s to a new set of operators $\{E_k\}$s such that the new error operators satisfy a modified form of the Knill-Laflamme (KL) conditions for the code  $\cC$.
\begin{align}\label{eq:ortho_Ek}
   \langle m_{\rm L}|E_k^{\dagger}E_\ell|n_{\rm L}\rangle &= \alpha_{k}^{mn} \delta_{k\ell}
\end{align}
This condition implies that different errors map the codespace to orthogonal subspaces, but the non-unitary deformation remains. We refer to Eq.~\eqref{eq:ortho_Ek} as the \emph{diagonal form} of the AQEC condition. This transformation is crucial because we can detect the different error syndromes without ambiguity before performing the recovery operation. We describe our algorithmic procedure in Sec.~\ref{sec:algorithm} below and the construction of the recovery operation in Sec.~\ref{sec:recovery}.  

\subsection{Towards the diagonal form of AQEC conditions}\label{sec:algorithm}
given a quantum channel $\cA$ with Kraus operators $\{A_{i}\}$ that satisfies the approximate QEC conditions in Eq.~\eqref{eq:beny_oreskov} for a code with projector $P$. To transform this into a channel that satisfies the diagonal form of the AQEC condition \emph{for the same codespace}, we identify the first element of the new Kraus set as
\begin{align}\label{eq:E1}
    E_1P = A_1P .
\end{align}
Here, $A_1$ is either the no-error operator (that is, the operator proportional to  identity) or the error operator that has the largest norm in terms of the nosie strength. For example, $A_{1}$ would be the no-damping operator for the amplitude-damping channel. 

Next, we obtain $E_2P$ from $A_2P$ as
\begin{align}\label{eq:E2}
    E_2P = A_2P - U_1 P_1 U_1^{\dagger} A_2P ,
\end{align}
where $A_{2}$ is chosen such that the noisy subspace generated by $A_2$ has minimal overlap with the noisy subspace generated by $A_1$. In Eq.~\eqref{eq:E2}, $U_1$ is the unitary from the polar decomposition
\begin{equation}
E_1P = U_1 \sqrt{P E_1^{\dagger} E_1 P}, \label{eq:polar_decomp}
\end{equation}
and $P_1$ is the projector onto the support of $P E_1^{\dagger} E_1 P$. 
The next Kraus operator $E_3P$ is constructed as
\begin{align}
    E_3P = A_3P - U_2 P_2 U_2^{\dagger} A_3P - U_1 P_1 U_1^{\dagger} A_3P ,
\end{align}
where $A_3$ is chosen so that the corresponding noisy subspace has the second-smallest overlap with the subspace generated by $A_1$. Here $P_2$ is the projector onto the non-null space of the operator $PE_2^{\dagger}E_2P$ and $U_2$ is the unitary from the polar decomposition of $E_2P$. 

We repeat this procedure iteratively to generate the rest of the operators of the set $\{E_kP\}$. The $k^{\rm th}$ Kraus operator of the new set is thus of the form,  
\begin{align}\label{eq:Ek_forms}
    E_kP &= A_kP - \sum\limits_{i=1}^{k-1} U_{i}P_iU^{\dagger}_iA_kP,
\end{align}
\begin{figure}
    \centering
    \includegraphics[width=1.0\columnwidth]{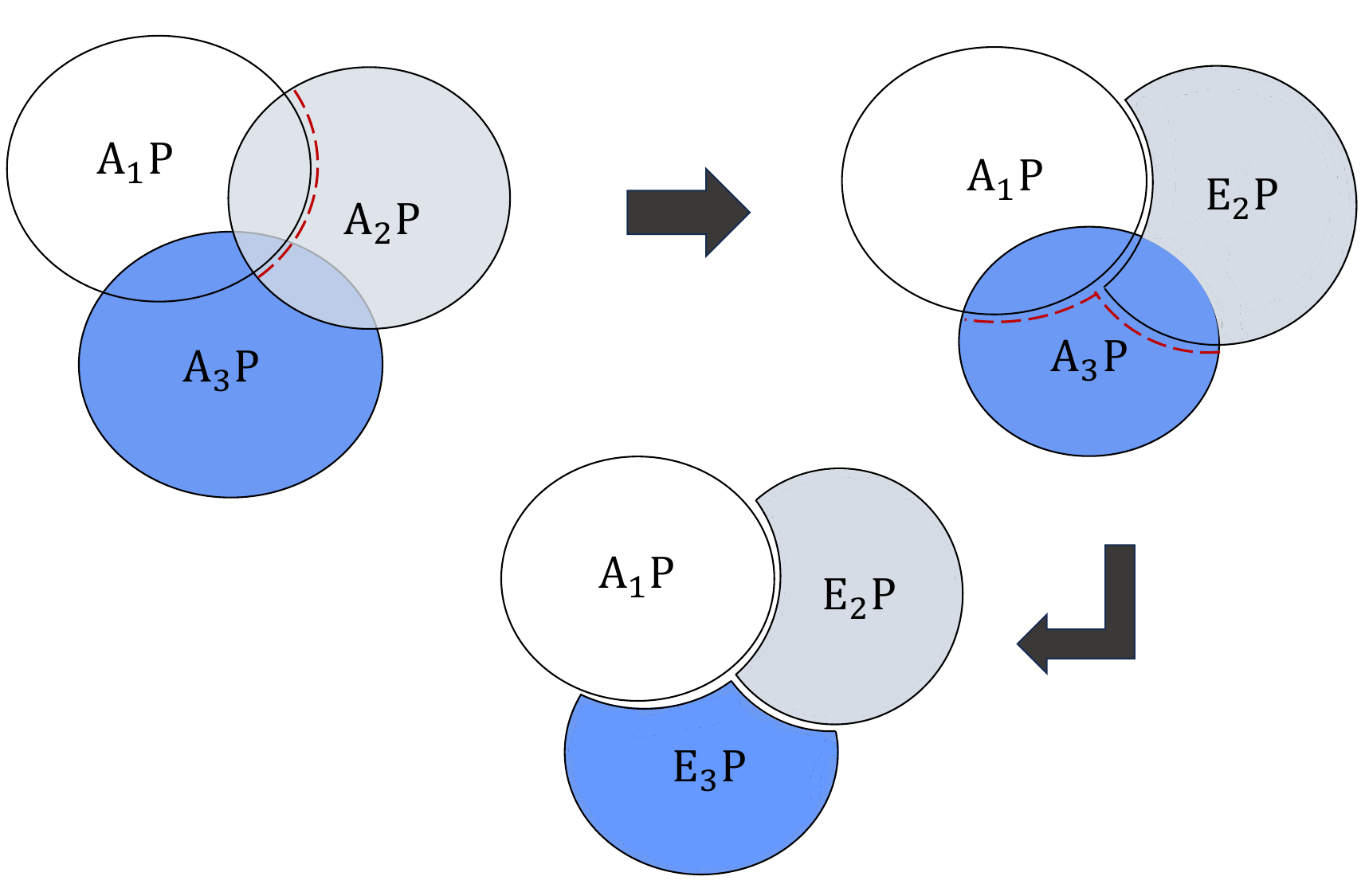}
    \caption{Schematics of the orthogonalisation algorithm, showing how the algorithm orthogonalises the overlapping subspaces. }
    \label{fig:ortho_schem}
\end{figure}
where the $P_i$s are the projector onto the non-null-space of the operators $PE_i^{\dagger}E_iP$, and $U_i$ are the unitary operators arising from the polar decomposition of the operators $E_kP = U_k \sqrt{PE_k^{\dagger}E_kP}$. A schematics of this procedure is shown in the Fig. \ref{fig:ortho_schem}.

It is worth noting that the new set of operators $\{E_k, k >1\}$ is related to the original set of Kraus operators $\{A_k, k>1\}$, as,
\begin{align}\label{eq:E_k}
    E_k P&=Q_kA_kP,
\end{align}
where $Q_k= I-\sum\limits_{i}U_{i}P_iU^{\dagger}_i$, and $Q_k$s are projectors (see the Appendix \ref{app:orth_channel}). We now show that the AQEC condition takes the diagonal form for this the new set of Kraus operators.  

\begin{theorem}\label{thm:PEkElP=0}
      The new set of Kraus operators in Eq.\eqref{eq:Ek_forms} satisfy the following orthogonality conditions.
     
    \begin{align}\label{eq: orthogonality}
    P_kE_k^{\dagger}E_\ell P_\ell= 0 \,\, \forall k\neq \ell
\end{align}
\end{theorem}
\begin{proof}
    We begin the proof by considering the first pairs of operators \{$E_1P_1$, $E_2P_2$\} from the newly constructed set of Kraus operators. The Eq.\eqref{eq:E2} and Eq.\eqref{eq:E1} gives the following 
    \begin{align} \label{eq:E1E2_line_1} 
        P_1E_1^{\dagger}E_2P_2 &= P_1A_1^{\dagger}A_2P_2 - P_1A_1^{\dagger}U_1P_1U_1^{\dagger}A_2P_2\\
    \label{eq:E1E2_line_2}    & =  P_1A_1^{\dagger}A_2P_2-  M_{11}^{1/2}U_1^{\dagger}U_1 P_1U_1^{\dagger}A_2P_2\\
   \label{eq:E1E2_line_3}       & =  P_1A_1^{\dagger}A_2P_2 - M_{11}^{1/2}P_1 M_{11}^{-1/2}P_1
   A_1^{\dagger}A_2P_2\\
       \label{eq:E1E2_line_4}   &=0 .
    \end{align}
    For the equality in the Eq.\eqref{eq:E1E2_line_2}, we consider the polar decomposition of $A_1P = E_1P = U_1\sqrt{PA_1^{\dagger}A_1P}$, the $M_{11}$ operator in Eq.\eqref{eq:E1E2_line_2}  is $M_{11}= PA_1^{\dagger}A_1P$  and also note that $E_1P_1 = U_1\sqrt{P_1A_1^{\dagger}A_1P_1}$. For the equality in the Eq.\eqref{eq:E1E2_line_3} we consider spectral decomposition of the operator $M_{11}$ first.
    \begin{align}\label{eq:spectral_decom_M_11}
        M_{11}= \sum\limits_{i=0}^{N-1} c_i |c_i\rangle\langle c_i|,
    \end{align}
    $\{c_i\}$ is the eigenvalues of $M_{11}$ and we keep the possibility that some $c_i$ may be zero and $N$ is the dimension of the operators $M_{11}$. Consider there are $N'$ numbers of non-zero $c_i$s in $\{c_i\}$. Therefore the $P_1$ is the following 
    \begin{align}
        P_1 &= \sum\limits_{i=0}^{N'-1} |c_i\rangle\langle c_i |.
    \end{align}
     Therefore the pseudo-inverse $M^{-1/2}_{11}$ is the following 
     \begin{align}
          M_{11}^{-1/2}= \sum\limits_{i=0}^{N'-1} d, _i^{-1/2} |c_i\rangle\langle c_i|
     \end{align}
     Therefore $M_{11}^{1/2}P_1M_{11}^{-1/2}$ is 
     \begin{align}
        M_{11}^{1/2}P_1M_{11}^{-1/2} & = \sum\limits_{i=0}^{N'-1} |c_i\rangle\langle c_i|=P_1. 
     \end{align}
     We also note that the $P_1E_1^{\dagger}E_2P_2$ implies that 
     \begin{align}
         PE_1^{\dagger}E_2P= 0.
     \end{align}

     Using the similar steps we can proceeds to prove the orthogonality for the pairs $\{E_2P_2,E_3P_3\}$ and $\{E_1P_1,E_3P_3\}$. Similarly the proofs of the orthogonality can be shown for any pairs of operators $  \{E_iP_i,E_{j}P_j \}$ and thus we can conclude $P_iE_i ^{\dagger}E_jP_j=0$ for all $i \neq j $.
 \end{proof}
As a consequence of the orthogonality of the operators $\{E_k\}$ we note that  that $P_iU_{i}^{\dagger}U_jP_j=0$ (See the Lemma \ref{lem:UkU_l=0}). $PE_k^{\dagger}E_{\ell}P=0$  signifies that the operators $\{E_kP\}$ maps the code space to different subspaces and this orthogonality is necessary to define the recovery operation in the following section. We outline a syndrome detection procedure in the Appendix \ref{app:uni_syn}. We will witness that the orthogonalisation algorithm not only paves a way to define a better recovery operation but also helps optimizing the circuit implementation of the recovery. The Fig.\ref{fig:ortho_sub} shows a schematic diagram of the syndrome subspaces orthogonalisation procedure. Note also, in general the number ($N_A$) of the actual noise operators $A_kP$ and the number ($N_E$) of the new error operators $E_kP$ are different. As a consequence of the orthogonalisation we also note that the QEC matrix $\tilde{M}_{kk}$ is less than the original QEC matrix $M_{kk}$ through the Lemma \ref{lem:M>Mt}.



\section{Syndrome-based Recovery using the Petz map }\label{sec:syn_petz}

At this juncture, we further extend the construction of the recovery operation with the help of newly constructed Kraus operators $\{E_k\}$s. We \emph{define} the Petz recovery, which is adapted to the channel defined by the Kraus operators $\{E_k\}$s. Therefore the for a code $\cC$ with projector $P$, the Petz recovery adapted to the channel $\cE \sim \{E_i\}$ looks as follows 
\begin{align}\label{eq:petz_w}
    \cR_{P,E}(.)&\equiv \sum\limits_{i=1}^{N}\, PE_i^{\dagger}\cE(P)^{-1/2}(.)\cE(P)^{-1/2}E_i P.
\end{align}
Note that here we define the operator $\cE(P)^{-1/2}$ on the support of $\cE(P)$. In the framework of QEC through Petz recovery, we generally adapt the Petz map to the noise process. We also note that from the Lemma \ref{lem:Ep<Ap} the $\cE(P)^{-1/2} \succeq \cA(P)^{-1/2}$.   However, we adapt the recovery to a different map $\cE$, which may not be TP, but as proved in the former section, it is trace non-increasing. We also note one crucial point: by the definition of Petz recovery, the map in the Eq.\eqref{eq:petz_w} is trace-preserving. Following the Ref. \cite{liang_jiang_near_opt}, we can write the analytical form of the Kraus operators of the Petz map in Eq.\eqref{eq:petz_w} as 
\begin{align}\label{eq:petz_kraus}
    R^{P,E}_{\ell} \equiv \sum\limits_{\mu,\nu} \Tilde{M}_{[\mu,\ell],[\nu,\ell]}^{-1/2} \,\,|\mu\rangle_{L}\langle \nu| E_\ell^{\dagger},
\end{align}
where $ \Tilde{M}_{[\mu,\ell],[\nu,k]} = \langle \mu |E_\ell^{\dagger}E_k|\nu\rangle$ are the QEC matrix elements concerning the newly formed Kraus operators in the Sec.\ref{sec:algorithm}. The detail calculation can be found in the Appendix \ref{app:kraus_sbpm}. Here also we note that the summation over the indices $\{\m,\n\}$ do not runs from $0- (d-1)$ as the support of $\tilde{M}_{\ell,k}$ is defined through only $N'$ numbers of code vectors.

\subsubsection{Fidelity bounds for the syndrome-based Petz recovery}\label{sec:opt_petz}

In general, a Petz map adapted to the noise process is near-optimal in terms of the measure of both the worst-case fidelity \cite{hkn_pm2010} and entanglement fidelity \cite{liang_jiang_near_opt}. It may also serve as \emph{optimal recovery} with regards to the measure of the entanglement fidelity if the Kraus operators of the noise process ${A_k}$s and the code $\cC$ satisfy the following commutation relation \cite{liang_jiang_optimality}
\begin{align}\label{eq:commutation_M}
     [M,\tr_{\rm L}(\sqrt{M}\otimes I_d)]=0,
\end{align}
where $M$ is called the QEC matrix and has the following matrix elements 
\begin{align}
    M := \{M_{[\mu,k],[\nu,\ell]} = \langle \mu|A_k^{\dagger}A_{\ell}|\nu\rangle\},
\end{align}
and $d$ is the dimension of the code $\cC$.
An arbitrary noise process and code generally do not satisfy the commutation in Eq.\eqref{eq:commutation_M}. 

In the previous section, we demonstrated a framework for an AQEC protocol by mapping back the orthogonal and non-overlapping subspaces, which we generated from the overlapping syndrome subspaces to the code space. The recovery in the previous section has been done through a unitary that appears from the polar decomposition of the operators $\{E_k\}$. However, the recovery is not optimal, as we see from the Eq.~\eqref{eq:fid_min_bound} that the lower bound on the fidelity is close to the bound of $F_{\rm min}^2$ for the Leung recovery. However, we can deploy the Petz map defined in the Eq.\eqref{eq:petz_w} to revert these orthogonal subspaces to the code space.

Now, by construction described in the Sec.\ref{sec:algorithm}, the operator $PE_k^{\dagger}E_lP=0$ for $k\neq l$. Therefore, the QEC matrix $\Tilde{M} = PE_k^{\dagger}E_lP$ is diagonal in the logical state basis.
Furthermore, the matrix $\sqrt{\Tilde{M}}$ has an block diagonal from in the logical basis. Hence, the partial trace with respect to the logical basis $\tr_L(\sqrt{\Tilde{M}}\otimes I_d)$ is diagonal. Therefore, the QEC matrix $\Tilde{M}$ satisfies the optimality condition given by the Eq. \eqref{eq:opt_qec_cond}, but for the noise channel $\cE$. The following theorem gives a bound on the entanglement fidelity of the syndrome-based Petz map in Eq.\eqref{eq:petz_w}, for the original noise channel $\cA$. 

 \begin{theorem}
      For any quantum code with projector $P$ onto the code space and a CPTP noise process $\cA$ the entanglement fidelity under the Petz map $\cR_{P,\cA}$ is bounded as follows
    \begin{align}\label{ineq:petz_syn_petz}
         F_{\rm ent}(\cR_{P,\cA}\circ \cA)& \geq F_{\rm ent}(\cR_{P,\cE}\circ \cA)^2.
    \end{align}
    Hence the fidelity losses $\eta_P = 1-  F_{\rm ent } (\cR_{P,\cA}\circ \cA)$ and $\eta_s = 1-  F_{\rm ent } (\cR_{P,\cE}\circ \cA)$ satisfy the following bound 
    \begin{align}
        \eta_P \leq 2 \eta_s 
    \end{align}
 \end{theorem}  
 \begin{proof}
     From \cite{barnum2002} we have 
     \begin{align}
         F_{\rm ent}(\cR_{\rm opt}\circ\cA)^2& \leq F_{\rm ent}(\cR_{P,\cA}\circ\cA).
     \end{align}
     Here $\cR_{\rm opt}$ is the optimal recovery map. Therefore any recovery including the $\cR_{P,\cE}$ the following inequality holds 
     \begin{align}
         F_{\rm ent } (\cR_{P,\cE}\circ \cA) \leq   F_{\rm ent } (\cR_{\rm opt}\circ \cA)
     \end{align}
     Therefore combining these two inequality we have the following 
     \begin{align}
        F_{\rm ent } (\cR_{P,\cE}\circ \cA)^2 \leq  F_{\rm ent } (\cR_{\rm opt}\circ \cA)^2\leq  F_{\rm ent } (\cR_{P,\cA}\circ \cA).
     \end{align}
     Therefore we have 
     \begin{align}
         (1- \eta_s)^2 &\leq 1- \eta_P\\
         2\eta_s &\geq \eta_P.
     \end{align}
     This completes the proof. 
 \end{proof} 

 Other than the near optimal performance of the the syndrome-based Petz map, we show in the Appendix~\ref{app:petz_vs_syn_petz}, the recovery $\cR_{P,\cE}$ posses the ability to outperform the the Petz map, under certain conditions. In the following section we will consider some approximate QEC codes, to demonstrate the performance of the syndrome-base Petz map.
 
\begin{figure}[t!]
    \centering
    \includegraphics[width=1\columnwidth]{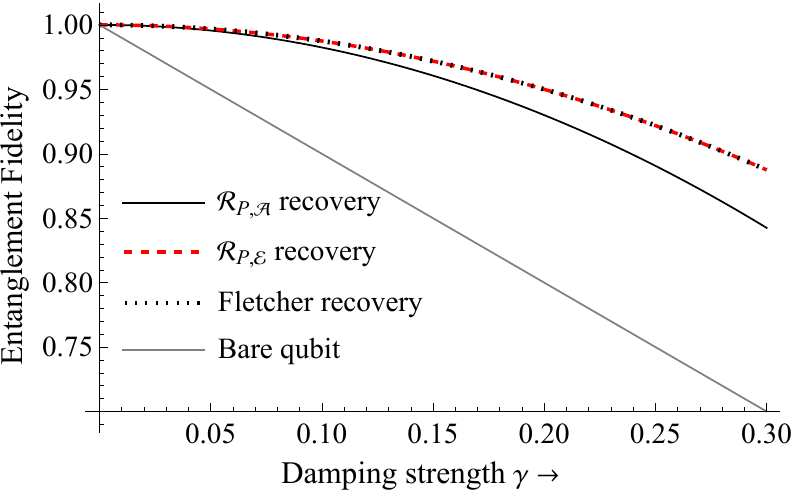}
    \caption{Comparison of the performance of various recovery protocols based on entanglement fidelity for the [[4,1]]-Leung code. Here the Fletcher recovery is the SDP-based recovery \cite{fletcher2008channel}. }
    \label{fig:ent_fids}
\end{figure}

\section{Examples: Syndrome-based AQEC protocols for amplitude-damping noise}\label{sec:ex}

To demonstrate the performance of our recovery operations, we consider the noise model to the amplitude damping noise. It is a dominating noise process in various hardware like superconducting processors \cite{chirolli2008decoherence}. The amplitude-damping (AD) noise channel has the operator-sum representation with the following Kraus operators \cite{chuang_nielsen}
\begin{align}
    D_0 &= \begin{pmatrix}
        1 & 0 \\ 0 & \sqrt{1-\gamma}
    \end{pmatrix} &  D_1 = \begin{pmatrix}
        0 & \sqrt{\gamma} \\ 0 & 0
    \end{pmatrix},
\end{align}
where the parameter $\gamma$ is the probability for the decay of the qubit from the excited state $|1\rangle$ to the ground state $|0\rangle$. Here, we refer to the parameter $\gamma$ as the damping strength or the strength of the AD noise.

 We consider two $[[4,1]]$ codes to correct the AD noise. One is the Leung code \cite{leung}, and the other is the optimal [[4,1]] code \cite{ak_good_code}. Note that the code in Ref.\cite{ak_good_code} is optimal when the recovery is fixed to the Petz map $\cR_{P,\cA}$.   
\begin{table}[t!]
    \centering
    \begin{tabular}{c c}
    \hline
    \hline
    Operators & Analytical form \\
    \hline
    \\
      $R_{0}$   & $|0_L\rangle (\alpha\langle 0000|+\beta \langle 1111|) +\frac{1}{\sqrt{2}}|1_L\rangle(\langle 0011|+ \langle 1100|) $  \\ \\

           $R_{1}$   & $|0_L\rangle|\langle 1110| +|1_L\rangle|\langle 0010| $  \\ \\
            $R_{2}$   & $|0_L\rangle|\langle 1101|+|1_L\rangle|\langle 0001| $  \\ \\
             $R_{3}$   & $|0_L\rangle|\langle 1011|+|1_L\rangle|\langle 1000| $  \\ \\ 
              $R_{4}$   & $|0_L\rangle|\langle 0111| +|1_L\rangle|\langle 0100| $  \\ \\
            $R_{5}$   & $|0_L\rangle|\langle 0110|$  \\ \\
            $R_{6}$   & $|0_L\rangle|\langle 1001|$  \\ \\
             $R_{7}$   & $|0_L\rangle|\langle 1010|$  \\ \\
              $R_{8}$   & $|0_L\rangle|\langle 1010|$  \\ \\
              $R_{9}$   & $|1_L\rangle(\beta\langle 0000|-\alpha \langle 1111|) +\frac{1}{\sqrt{2}}|0_L\rangle(\langle 0011|- \langle 1100|) $ \\ \\
              \hline
              \hline
    \end{tabular}
     \caption{Analytical form of the recovery operators $\cR_{P,\cE}$ for the [[4,1]] code. We see that these recovery operators are similar to the SDP-based recovery \cite{fletcher2008channel} operators with only difference in $R_9$. In the Fletcher recovery $|1_L\rangle$ is associated with the last term of the $R_9$ in this table. }
    \label{tab:rec_tab}
\end{table}


\subsection{The Leung [[4,1]] code}
To protect a single qubit of information against the amplitude-damping noise, we consider the following four qubit codes proposed by Leung \emph{et al.}\cite{leung}  
\begin{align}\label{eq:leung}
    |0_L\rangle &= \frac{1}{\sqrt{2}}|0000\rangle+|1111\rangle\\
    |1_L\rangle &= \frac{1}{\sqrt{2}}|0011\rangle+|1100\rangle.
\end{align}

Note that all the single-qubit damping error operators $D_{0001},D_{0010},D_{0100},D_{1000}$ acting on the code space, map the $\cC$ to orthogonal subspaces~\cite{leung}, where $D_{ijkl}= D_i\otimes D_j \otimes D_k \otimes D_l$. Therefore, our operator orthogonalisation algorithm would remain ineffective for this code. The details of the construction of the new kraus operators for the code in Eq.\eqref{eq:leung} are in the Appendix \ref{app:new_kraus_cons}.  Hence, performing the recovery proposed in the Sec.\ref{sec:recovery} will result in a similar fidelity to the Leung recovery. However, there is a slight difference between them. The difference is that the recovery obtained from the orthogonalisation can also correct the two-qubit errors to some extent, but the Leung recovery cannot do that. Because of this reason, we have different curves in the plots in Fig.~\ref{fig:ent_fids} for the Leung recovery and the $R_E$-recovery. 


When we deal with the syndrome-based Petz recovery, we encounter a normaliser map $\cE(P)^{-1/2}(.)\cE(P)^{-1/2}$. We also have this in the standard Petz recovery. This operation helps reduce codespace deformation due to the perturbative satisfaction of the KL condition. Therefore, even if the syndrome spaces are orthogonal, we can recover better using a normal Petz map than Leung recovery. Now, in our work, we consider a syndrome-based version of the Petz map. This syndrome-based Petz map reverts the modified subspace obtained from the orthogonalisation to the codespace. It not only helps reduce deformation in the codespace but also partially corrects two-qubit errors. We argued in Sec.~\ref{sec:opt_petz} that the syndrome-based Petz map generates entanglement fidelity $F_{\rm ent}$ similar to the recovery noted in \cite{fletcher2008channel}. We should note that the Kraus operators of the syndrome-based Petz map differ from those of the optimal recovery obtained via semi-definite programming. The Table \ref{tab:rec_tab} shows the Kraus operator of the recovery $\cR_{P, E}$.

We compare the worst-case fidelity achieved through the recovery protocol $\cR_E$ with two other noise-adapted recovery operations; one is the Leung recovery \cite{leung}, and the other is the Petz recovery (Eq.~\eqref{eq:petz_def}), which serves as the near-optimal recovery in the Table \ref{tab:wcf_table}. We note that the recovery performance of $\cR_E$ is better than the Leung recovery, but the Petz recovery outperforms $\cR_E$. To this end, we also compare the performance of the Petz map $\cR_{P, E}$ with the other recovery protocols in terms of entanglement fidelity in Figure \ref{fig:ent_fids}, for the $[[4,1]]$ code, where we see that the entanglement fidelity under the recovery map $\cR_{P, E}$ is

\begin{align}
    F_{\rm Ent}(\cR_{P,\cE}\circ \cA) &= 1- 1.25 \gamma^2 +\cO(\gamma^3),
\end{align}

which it greater that the $ F_{\rm Ent}(\cR_{P,\cA}\circ \cA)$ for the four-qubit Leung code. Moreover, we note that  $F_{\rm Ent}(\cR_{P,\cE}\circ \cA)$ with [[4,1]]-Leung code is similar to the Fletcher optimal recovery \cite{fletcher2008channel} up to the leading order of the noise strength. 

\begin{widetext}

\begin{table}[t!]
    \centering
    
    \begin{tabular}{|c||c|c||c|c|}
     \hline 
     & \multicolumn{2}{c||}{Four-qubit Leung code \cite{leung}} & \multicolumn{2}{c|}{$[[4,1]]$-structured code }\\
     \hline
       Recovery   & Worst case Fidelity & Entanglement fidelity  & Entanglement Fidelity & Worst case fidelity  \\
           \hline
      Leung   & $1-2.75\gamma^2$ & $1-3\gamma^2$  &  $\times$& $\times$ \\
      \hline
      Petz $\cR_{P,\cA}$   & $1-1.75\gamma^2$ & $1-1.75\gamma^2$ & $1- 1.86 \gamma^2$& $1- 1.33 \gamma^2$\\
       \hline
      $\cR_E$   & $1-2.47\gamma^2$ & $1-1.996\gamma^2$ &$1-2.04\gamma^2$ &$1-2.1\gamma^2$\\
       \hline
      $\cR_{P,\cE}$  & $1-1.15\gamma^2$ &$1- 1.25 \gamma^2$ & $1- 1.53\gamma^2$ &$1-1.134\gamma^2$\\
       \hline
    \end{tabular}
    \caption{ Fidelity expressions for different recovery operations upto $\cO(\gamma^2)$. We consider the two types of encoding. One is the [[4,1]]-Leung code, and the other is also a [[4,1]] code, which serves as optimal code if the recovery is the Petz map \cite{ak_good_code}. The ``$\times$" indicates that Leung recovery is not applicable for the [[4,1]]-structured code. All these expression has been obtained by fitting the numerically obtained data with the polynomial $1- a_1 \gamma - a_2 \gamma^2 - a_3 \gamma^3 - a_4 \gamma^4 - a_5\gamma^5$. After the fitting we obtain the numerical value of $a_1$ is vanishingly small and of the order of $\sim 10^{-6}$ for all the recovery and all the codes.  }
    \label{tab:wcf_table}
\end{table}
\end{widetext}


\subsection{Other four-qubit codes}
To see the performance of our recovery protocols outlined in the Secs.\ref{sec:recovery}- \ref{sec:syn_petz} for the other four qubit codes, we consider the four qubits optimal code \cite{ak_good_code}, obtained from a structured as well as unstructured search in the 4-qubit Hilbert space by fixing the recovery operation to the Petz map. We also consider another four-qubit code, which performs optimally under a optimal recovery (obtained through a numerical search) in the Ref. \cite{liang_jiang_opt_code}.
\subsubsection{[[4,1]]-structured numerical code}
From the paper \cite{ak_good_code} we choose the code from the structured search and we explore the performance of different recovery protocols. We note that the requirement for the Leung recovery is that the syndrome subspaces should be orthogonal  (See Eq. (29) in \cite{leung}). This orthogonality is also a requirement for the Cafaro recovery \cite{cafaro_PhysRevA.89.022316}. However, none of the four qubit codes in \cite{ak_good_code} satisfy this orthogonality criterion. Therefore, the codes in \cite{ak_good_code} do not admit the Leung or Cafaro recovery. 

However, our work solves this problem by orthogonalising the syndrome subspaces. To demonstrate the performances of the orthogonalisation procedure, we consider the code obtained from the "structured search" (see Sec. III of \cite{ak_good_code}). We first proceed with the recovery map $\cR_E$. The Entanglement fidelity and the worst-case fidelity are the following. 

\begin{align}
F_{\rm Ent} (\cR_{E}\circ\cA)  &= 1- 2.04 \gamma^2 +\cO(\gamma^3) \\
F^2_{\rm min}(\cR_{E}\circ\cA)  &= 1- 2.1 \gamma^2 +\cO(\gamma^3)
\end{align}

For the syndrome-based Petz map $\cR_{P, E}$, the entanglement and the worst-case fidelities are the following.

\begin{align}
    F_{\rm Ent}(\cR_{P,\cE}\circ\cA) &= 1- 1.53\gamma^2 +\cO(\gamma^3)\\
    F^2_{\rm min}(\cR_{P,\cE}\circ\cA)  &= 1- 0.52\gamma^2 +\cO(\gamma^3)
\end{align}

We, therefore, note that from Table \ref{tab:wcf_table}, the performance of the syndrome-based Petz map is much better than the normal Petz map for this code. 

\subsubsection{Optimal [[4,1]]-code from biconvex search}
Here we consider the following [[4,1]] AQEC code from the article \cite{liang_jiang_opt_code}, which is tailor made to correct the amplitude-damping noise 
\begin{align}\label{eq:liang_code}
    |0_L\rangle& = \sqrt{1- \frac{1}{2(1-\gamma^2)}}|0000\rangle +\frac{1}{\sqrt{2}(1-\gamma)}|1111\rangle\\
    |1_L\rangle&= \frac{1}{2}(|0011\rangle+|1100\rangle +|0101\rangle-|1010\rangle).
\end{align}
Note that the encoding depends on the value of the noise strength or the damping strength $\gamma$. This code results from a biconvex search that aims to achieve an optimal performance against the amplitude-damping noise. To do so, the author of \cite{liang_jiang_opt_code} has also figured out the recovery through the biconvex optimisation. Not only that, but there exists another recovery that is engineered analytically. The numerically optimized recovery results in a entanglement fidelity around $F_{\rm ent} = 1- 1.09\gamma^2 +\cO(\gamma^3)$, which the analytically obtained recovery in \cite{liang_jiang_near_opt} gives $F_{\rm ent} = 1- 1.85\gamma^2 +\cO(\gamma^3)$. 
\begin{figure}[t!]
    \centering
    \includegraphics[width=1\columnwidth]{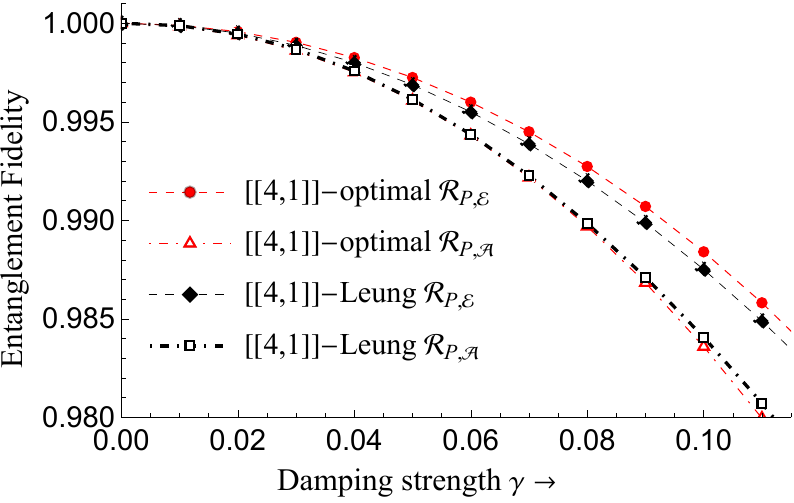}
    \caption{Performances of different codes under the Petz and the syndrome-based Petz map. Here the [[4,1]]-optimal $\cR_{P,\cE}$ refers to the syndrome-based Petz recovery with the optimal [[4,1]]-code in Eq.\eqref{eq:liang_code}. The $\cR_{P,\cA}$ is the ordinary Petz map, and the $R_{P,\cE}$ is the syndrome-based Petz map. We also note the under the Petz recovery $\cR_{P,\cA}$ the entanglement fidelity is $F_{\rm ent} = 1- 1.5 \gamma^2$, which is better than the Leung code with $\cR_{P,\cA}$ (see the Table \ref{tab:wcf_table}). }
    \label{fig:opt_code_syn_petz}
\end{figure}
In our paper, we consider the code in Eq.\eqref{eq:liang_code}, then construct the new set of Kraus operators through our subspaces orthogonalisation procedure, and then apply the syndrome-based Petz recovery $R_{P,\cE}$. The detailed construction of the new Kraus operators for the code in Eq.\eqref{eq:liang_code} is outlined in Appendix \ref{app:new_kraus_cons}. We show the performance of the syndrome-based recovery for the code in Eq.\eqref{eq:liang_code} in Fig. \ref{fig:opt_code_syn_petz}. It shows that the optimal code and the syndrome-based Petz outperform the syndrome-based Petz with the [[4,1]] Leung code. We have discussed the syndrome-based Petz with the $[[4,1]]$ code, which performs similarly to the Fletcher optimal recovery. Therefore, the optimal code in Eq.\eqref{eq:liang_code} with $\cR_{P,\cE}$ outperforms the Fletcher optimal recovery  with [[4,1]]-Leung code also. We further fit the numerically obtained data in the Fig.\ref{fig:opt_code_syn_petz} with a fifth-order polynomial $1- \sum\limits_{i =1}^{5} a_i \gamma^i$. The curve fitting shows that the entanglement fidelity $F_{\rm ent} \approx 1- 1.05 \gamma^2 + \cO(\gamma^3)$, which is similar to the reported entanglement fidelity in \cite{liang_jiang_opt_code}. We obtained this optimal fidelity by choosing the new Kraus operators $\{E_kP\}$ in a specific manner. Otherwise, we could have a different fidelity as described in the Appendix \ref{app:new_kraus_cons}
Along with the optimal performance, the code in Eq.\eqref{eq:liang_code} and recovery obtained through the numerical search or our algorithm have a drawback. We note that the recovery operators, the numerically obtained and our syndrome-based Petz, depend on the noise strength. However, to implement the recovery, we can fix the noise-strength dependent parameters in the recovery and the code, as we know the qubit specification on the hardware. However, on the hardware, each qubit has its own separate $T_1$; hence, the $\gamma$, the damping parameters, are different. However, to implement the QEC with the code in Eq.\eqref{eq:liang_code}, we can stick to an average value of the $\gamma$ of the qubits participating in the encoding operation. We can also fix the parameters for the recovery operation (syndrome-based Petz) according to the encoding parameters. However, there will be a trade-off. We will lose the optimality in this practice. We can bypass these difficulties with the Leung code and the syndrome-based Petz, which we adapted to it. There we can see from the Table \ref{tab:rec_tab} that, at least in estimating the logical $T_1$, i.e, to estimate the improved $T_1$ after the QEC, the QEC circuit is independent of $\gamma$ -- the noise strength.       
  
\subsection{Six-qubit code }
So far we have dealt with the amplitude-damping noise and the noise-adapted code. Now we demonstrate the performance of a perfect degenerate code under Pauli noise, namely the depolarizing noise, which has the following Kraus operators
\begin{align}
A_0 = \sqrt{1-p}I, A_1=\sqrt{\frac{p}{3}}X,A_2=\sqrt{\frac{p}{3}}Y,A_3=\sqrt{\frac{p}{3}}Z,
\end{align}
where the $\{X,Y,Z\}$ are the Pauli matrices.

To correct the depolarizing error, we consider a degenerate [[6,1,3]] code \cite{wilde_six} with the following stabilizer generators
\begin{align}
\cS = \langle YIZXXY,ZXIIXZ,IZXXXX,\nonumber\\IIIZIZ,ZZZIZI\rangle.
\end{align}

Note that the six-qubit code \cite{wilde_six} has a distance $d=3$, implying that it can perfectly correct all weight-one Pauli errors. We also note that the dimension of the syndrome space is $2^{n-k}= 32$. Therefore, apart from the eighteen weight-one Pauli errors, we can correct some weight-two Pauli errors. Using our orthogonalisation algorithm, we find that the following two-qubit errors are also correctable
\begin{align}\label{eq:two_wight_error}
\{X_5X_6,X_4X_6,X_3X_6,X_2X_6,X_1X_6,X_3X_5,\nonumber\\
X_3X_4,X_2X_4,X_1X_4,X_1X_2,Z_1Z_5,Z_2Z_6,Z_2Z_4\}.
\end{align}

Therefore, one can notice that, including the identity, weight-one, and weight-two errors, we have a total of thirty-two errors which can be corrected through the syndrome-based Petz. We demonstrate the performance of the syndrome-based Petz under depolarizing noise in Fig.~\ref{fig:six_depol}. We see that the stabilizer-measurement-based recovery, i.e., the standard recovery, performs similarly to the syndrome-based Petz, while the Petz recovery $\cR_{P,\cA}$ (here the $\cA$ contains all the noise operators ranging from identity to weight-six) performs worse than the other two recoveries.   
\begin{figure}
    \centering
    \includegraphics[width=1.0\columnwidth]{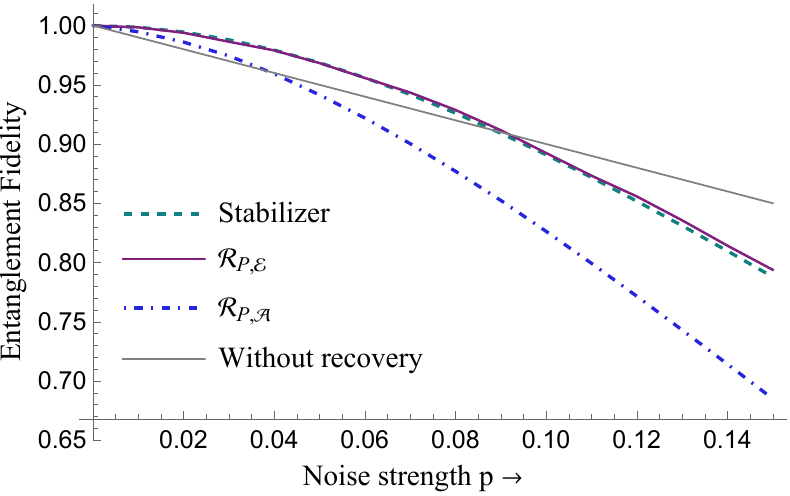}
    \caption{Performance of different recovery operation for the [[6,1,3]] code under the depolarizing noise process.}
    \label{fig:six_depol}
\end{figure}
However, from Fig. \ref{fig:six_amp} we note that the Petz map with the degenerate [[6,1,3]] code performs better than the syndrome-based Petz map under the amplitude-damping noise .

\begin{figure}
    \centering
    \includegraphics[width=1.0\columnwidth]{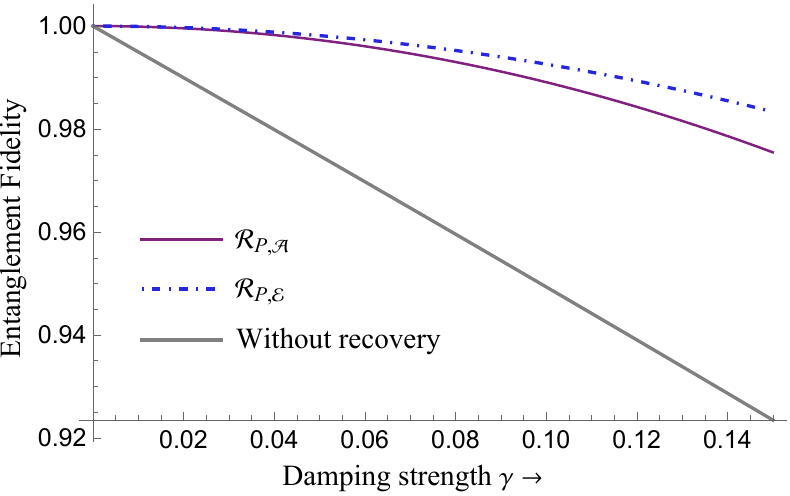}
    \caption{Performance of different recovery with the six-qubit degenerate code under the amplitude-damping noise process. }
    \label{fig:six_amp}
\end{figure}

\section{Circuit implementation of the syndrome-based Petz recovery}\label{sec:implementation}

Here in this section, we make use of our orthogonalisation procedure to design the circuit for the Petz recovery map and show that the circuit can correct the single-qubit errors on the hardware.

To begin with, we first note that there are other possible methods for designing the circuits for a Petz recovery or any CPTP recovery operation. One of the pioneering works is based on the \emph{QSVT} algorithm by Gilliyen \emph{et al.} \cite{gilyen2022_petz} but the resource complexity is exponentially high. However, there exist other methods \cite{biswas2024noise} to implement noise-adapted recovery protocols in a circuit but because of their high resource requirements, the protocols cannot demonstrate the QEC on the hardware. Notably other approaches of implementing a single qubit Petz map exists by Wen-Han \emph{et al.} \cite{vscarani_petz} on the trapped ion processor and the recent article shows the implementation on \textsc{NMR} processor \cite{iiser_mohali_petz}. However, all these experimental implementations is for the state-specific Petz map [see the Chap. 10 of \cite{mm_wilde}] . There has been other approaches based on the amplitude amplification and QSVT for implementing a \emph{Petz -like} map \cite{Petz_like}, but again the resources are too high to be implementable on a physical hardware. One of the notable generalisation of the Petz map have been proposed by Basak \emph{et al.} \cite{basak2025approximate} in the context of dynamic AQEC code. However, the usefulness of the Petz map on dealing with the hardware noise and an efficient implementation of it in the context of QEC is missing from the literature. Before proceeding with our implementation we stress the fact that some well established methods \cite{iiser_mohali_2022} for designing the circuit for a CPTP map can also be used to implement Petz map, but again to deal with larger number of qubit will cost huge circuit complexity. In our work we particularly focus on the QEC with the syndrome based Petz map based on the [[4,1]]-Leung code.  

Here we bring in the syndrome detection methods through stabilizer measurements in our circuit implementation for the syndrome-based Petz map in order to perform the noise adapted QEC on the hardware. Note that in general for a noise-adapted recovery, the syndrome detection is not possible but the algorithm in the Sec.\ref{sec:algorithm} makes syndrome detection possible for arbitrary quantum codes. Since we have chosen a superconducting hardware for demonstrating the noise adapted error correction, we recall that amplitude-damping noise is the dominating noise process in the hardware. Therefore, we consider the [[4,1]] code in Eq.\eqref{eq:leung} to correct single qubit amplitude-damping errors. For the recovery procedure we consider  the syndrome-based Petz recovery $\cR_{P, E}$ adapted to the four-qubit code in Eq.\eqref{eq:leung} and the amplitude-damping noise. Kraus operators for the syndrome-based Petz $\cR_{P, E}$ are listed in the Table~\ref{tab:rec_tab}. In our circuit implementation, we do not implement the full recovery, with which we can correct some two-qubit errors also, instead, we construct a map $\Tilde{\cR}$ with the Kraus operators $\tilde{\cR}\sim\{R_0, R_1, R_2, R_3, R_4\}$. The $\Tilde{\cR}$ can correct only the single-qubit amplitude-damping errors $ \{D_{1000},D_{01000},D_{0010},D_{0001}\}$ and the no-damping effect $D_{0000}$, approximately.

To proceed with the implementation we first note that the Kraus operators of the map $\Tilde{\cR}$ have the following form 
\begin{align}
    R_k &= G_k \Pi_k,
\end{align}

Where $\{\Pi_k= \sum\limits_{\m}|\psi_{[\m k]}\rangle\langle|\psi_{[\m k]}\}$s are the projector onto the syndrome subspaces, here $|\psi_{[\m k]}\rangle$ are defined in the Appendix \ref{app:kraus_sbpm} and $G_k$s are the unitary though which we execute the recovery. Therefore the recovery $\Tilde{\cR}$ can be implemented on a circuit by measuring the syndrome first and then followed by applying a recovery through the unitary operator $G_k$ according to the syndrome outcome. 
 \begin{figure}[t]
     \centering
     \begin{subfigure}[t]{0.3\textwidth}
         \centering
         \includegraphics[width=\textwidth]{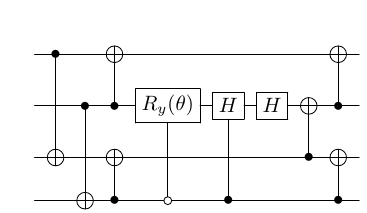}
         \caption{Circuit for the Unitary $G_0$ }
         \label{fig:uni0}
     \end{subfigure}
     \begin{subfigure}[t]{0.175\textwidth}
         \centering
         \includegraphics[width=\textwidth]{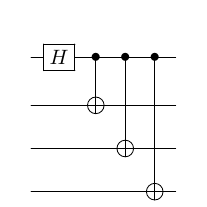}
         \caption{Circuit for the Unitary $G_1$ }
         \label{fig:uni1}
     \end{subfigure}
     \begin{subfigure}[t]{0.15\textwidth}
         \centering
         \includegraphics[width=\textwidth]{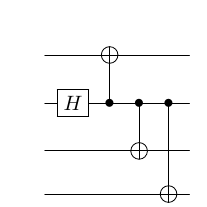}
        \caption{Circuit for the Unitary $G_2$ }
         \label{fig:uni3}
     \end{subfigure}
       \begin{subfigure}[t]{0.15\textwidth}
         \centering
         \includegraphics[width=\textwidth]{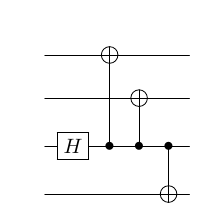}
       \caption{Circuit for the Unitary $G_3$ }
         \label{fig:uni4}
     \end{subfigure}
       \begin{subfigure}[t]{0.15\textwidth}
         \centering
         \includegraphics[width=\textwidth]{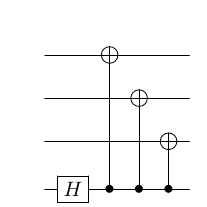}
        \caption{Circuit for the Unitary $G_4$ }
         \label{fig:uni5}
     \end{subfigure}
        \caption{Circuits for the Unitary matrices $G_k$s. }
        \label{fig:Uks}
\end{figure}
To detect the accurate syndrome subspaces, we first measure the stabilizer generators $ \cS^1_z =\langle  ZZII,IIZZ \rangle$. We denote $p_1$ and $p_2$ as the outcome of the measurement of $ZZII$ and $IIZZ$ respectively and call them the \emph{primary syndrome}. Since the primary-syndromes are incapable of concluding a deterministic syndrome extraction, meaning we cannot conclude which qubit has faced the error, we perform secondary syndrome checks by measuring the operators $\cS_z^2 = \langle ZIII, IIIZ\rangle$. Note that the operators in $\cS_z^2$ do not stabilise the code-space but they  stabilize the noisy subspaces. The outcomes of the complete syndrome measurements (primary + secondary) are listed in Table \ref{tab:syn_tabl}. In that table we should note that while applying the recovery $G_0$ we measure only the primary syndromes and measuring the secondary syndromes are not necessary. We also note that the secondary checks are conditioned on the primary syndrome, meaning if any of the primary syndromes ($p_1$ or $p_2$) are non-zero then only the secondary syndrome extractions are executed. The complete circuit for the syndrome extraction is shown in the Fig.\ref{fig:full_rec}.    

\begin{table}[t!]
    \centering
    \begin{tabular}{|c||c  c | c c || c|}
    \hline
      &\multicolumn{2}{c|}{$\cS_z^1$} &\multicolumn{2}{c||}{$\cS_z^2$} & \\
      \hline
    Noise     & $ZZII \, (p_1)$  & $IIZZ \, (p_2)$ & $ZIII \, (s_1)$ & $IIIZ \, (s_2)$ &  Recovery\\
    \hline
    $D_{0000}$     & 0 &0 & $\times$ &$\times$ & $G_{0}$\\
    $D_{1000}$     & 1 &0 & 1&0 & $G_1$\\
    $D_{0100}$     & 1 &0  & 0 &0 & $G_2$\\
    $D_{0010}$     & 0& 1 & 0 &0 & $G_3$\\
    $D_{0001}$     &0 &1  &0& 1& $G_4$\\
    \hline
    \end{tabular}
    \caption{Syndrome table for the $[[4,1]]$ code under amplitude-damping noise. Here, the parameters $p_i$ and $s_i$ denote the measurement outcome of the operators in $\cS_z^1$ and $\cS_z^2$, respectively. We do not conduct the extraction of the secondary syndrome if the $(p_1, p_2)=(0,0)$, that is what the $\times$ symbol represents. }
    \label{tab:syn_tabl}
\end{table}

Now we note that each unitary $G_k$ are 4-qubit unitary matrices. So to implement an arbitrary $4$-qubit unitary we at least $256$ single+ two - qubit gates. But here we will not be implementing the whole unitary $G_k$.   
Therefore, to implement the unitaries $G_k$s we first focus on the action of the $G_{k}$s on the noisy states. The $G_k$ takes the noisy state as an input and maps it to the logical code vectors. For example the unitary $G_1$ produces $|0_L\rangle \, or \, |1_L\rangle$ if the inputs are $|1110\rangle \,\, or \,\, |0010\rangle$ respectively. If the inputs are different from the vectors $|1110\rangle \,\, and \,\, |0010\rangle$, the unitary will not be effective. These selective actions of the unitary operators, meaning the $G_1$  (or in general $G_k$) will act on a perticular set of vectors and the rest of the vectors in the 4-qubit Hilbert space will remaining unaffected by the action of $G_k$,  are ensured by the syndrome measurements. For example, if we detect any of the noisy states $\{|1110\rangle, |0010\rangle\}$, we thorough the syndrome extraction, then only we apply the $G_2$. These combined operations (first syndrome then recovery) ensure that the $G_2$ will act only on the subspace spanned by $|1110\rangle \,\, and \,\, |0010\rangle$. Similarly if we detect $\{|1101\rangle \,\,{\rm or}\,|0001\rangle\}$ we apply $G_3$ and $G_4$ will be effective if and only if any of the states $\{|1110\rangle \,\,{\rm or}\,|0010\rangle\}$ are detected.  Fig.~\ref{fig:Uks} shows the circuit for the unitaries $G_k$s and the complete recovery operation respectively.
Here, one should note that a similar description---based on syndrome detection followed by applying the suitable unitary---for designing the circuit for the Petz map ($\mathcal{R}_{P,\mathcal{A}}$) can also be proposed. However, in that approach, the unitaries $G_i$ for $\mathcal{R}_{P,\mathcal{A}}$ would depend on the noise strength $\gamma$. Our orthogonalisation algorithm makes it possible to obtain a description of the Petz map in which the noise-strength--dependent unitary appears only in $G_0$, which is essential for handling the no-damping error $D_{0000}$.

In the hardware run, for correcting the $|1_L\rangle$, we refrain ourselves from applying the circuits in Fig.\ref{fig:Uks}. Because the, on the hardware we do not have all to all qubit connectivity.     
Therefore, the transpilation of our circuit on the hardware qubit connective will increase the circuit complexity, i.e, the gate counts and the the circuit depth as well for the circuit in Fig.\ref{fig:full_rec}.  In the hardware run, we consider only the \emph{linear-nearest neighbour} qubit connectivity. Since we do not have an all-to-all qubit connection the transpilation of the circuits for the unitary $G_i$s on the hardware qubit map increases the two-qubit gate counts to \emph{forty-seven} while the circuit in the Fig.~\ref{fig:full_rec} has only \emph{twenty-seven}  two-qubit gates. The increment of the two-qubit gates can deteriorate the performance of the recovery circuit in the hardware since the two-qubit gates are more noisy than the single-qubit gates and a faulty two-qubit operation causes the error to propagate to the other data qubits. Therefore reduction of the two-qubit gates is the key requirement for a successful demonstrating the error correction on the hardware.

\begin{widetext}

\begin{figure}[t!]
    \centering
    \includegraphics[width=1.0\columnwidth]{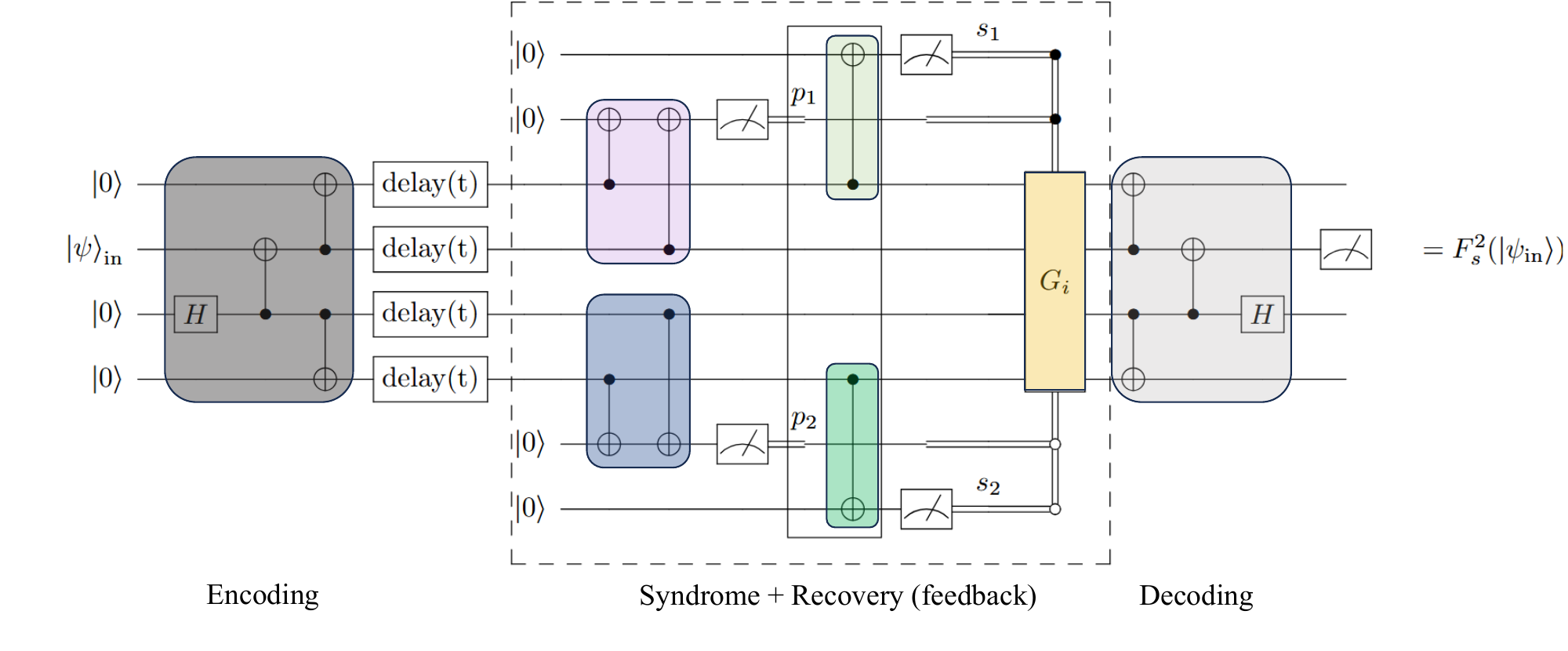}
    \caption{Combined circuit for the encoding + syndrome measurement + recovery + decoding (inverse of encoding). For the $T_1$ experiment we initiate the $|\psi\rangle_{\rm in}$  as $|1\rangle$. To carry out the $T_1$ experiment on the IBM processors we vary the delay time $t$ on the delay gates in the units of $\mu s$. To estimate the fidelity between the $|\psi\rangle$ and the recovered state $\cR_{P,\cE}\circ \cA(|\psi\Rl)$ we measure the second qubit in the $Z$- basis. A detail calculation for the fidelity estimation for input state $|\psi\rangle$=$|1\rangle$ is shown in the Appendix \ref{app:fid_estimation} }
    \label{fig:full_rec}
    
\end{figure}

    
\end{widetext}
\begin{table}[t]
    \centering
    \begin{tabular}{||l | l |l | l ||}
    \hline
    & & & \\
    Hardware & Qubits for the Experiment & $T_1$ and $T_2$ for the input qubit & Max and Min.\ two-qubit errors \\      
    \hline      
    & & & \\
    IBM-Brisbane &
      \parbox[t]{10em}{Q0-Q1-Q2-Q3-Q4-Q5-Q6-Q7} &
      \parbox[t]{10em}{$337\,\mu s$ and $272.54\,\mu s$} &
      \parbox[t]{12em}{$8.695\times10^{-3}$ and $3.604\times10^{-3}$} \\
    \hline
    & & & \\
    IBM-Torino (Fake backend) &
      \parbox[t]{10em}{Q1-Q2-Q3-Q4-Q5-Q6-Q7-Q8} &
      \parbox[t]{10em}{$155\,\mu s$ and $272.54\,\mu s$} &
      \parbox[t]{12em}{$5.695\times10^{-3}$ and $3.604\times10^{-3}$} \\
    \hline
  \end{tabular}
    \caption{Specification of the backends (real and copied). The copied backend means that we can copy the noise model onto a local simulator in Qiskit. We choose a particular set of qubits such that the two-qubit errors remain below $10^{-2}$. One more crucial point to note is that on Brisbane a single cycle of QEC takes $10\,\mu s$, while on Torino it takes $2.014\,\mu s$. For both Brisbane and Torino, the average readout error is $\sim10^{-2}$.}
    \label{tab:tab_fidelity_1}
\end{table}
However, we note that if we focus on recovering a particular logical state rather than a generic logical state $|\psi_L\rangle= a |0_L\rangle+b |1_L\rangle$, the circuit complexity of the recovery circuit gets reduced further. For example, if we focus on the $|1_L\rangle$ for the Leung code, the Fig.~\ref{fig:Uk-1s} indicates that after transpilation, the two-qubit gate counts for the recovery module will not increase as the circuits for the $G_k$s recovering the $|1_L\rangle$ state are already satisfying the linear qubit connectivity of the hardware. However, the transpiled version of the syndrome detection will consist of more two-qubit gates in comparison to the circuit in Fig.~\ref{fig:full_rec} as we have some long-range entangling operations like \textsc{CNOT}s.  In Table \ref{tab:gates_tab} we compare the total two-qubit and single-qubit gate counts for both the transpiled and un-transpiled circuits for the combined circuit Fig.~\ref{fig:full_rec}. Table \ref{tab:gates_tab} also shows that the circuit that recovers the $|1_L\rangle$ has much fewer two-qubit gates than the generic state recovery circuit.
 \begin{figure}[t]
     \centering
     \begin{subfigure}[t]{0.175\textwidth}
         \centering
         \includegraphics[width=\textwidth]{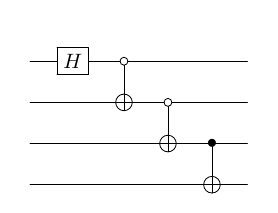}
         \caption{Circuit for the Unitary $G_1$ for $|1_L\rangle$. }
         \label{fig:uni1-1}
     \end{subfigure}
     \begin{subfigure}[t]{0.15\textwidth}
         \centering
         \includegraphics[width=\textwidth]{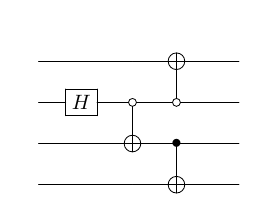}
        \caption{Circuit for the Unitary $G_2$ for $|1_L\rangle$.}
         \label{fig:uni2-1}
     \end{subfigure}
       \begin{subfigure}[t]{0.15\textwidth}
         \centering
         \includegraphics[width=\textwidth]{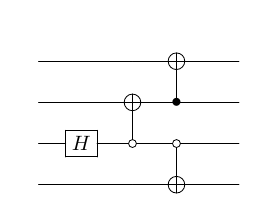}
       \caption{Circuit for the Unitary $G_3$ for $|1_L\rangle$. }
         \label{fig:uni3-1}
     \end{subfigure}
       \begin{subfigure}[t]{0.15\textwidth}
         \centering
         \includegraphics[width=\textwidth]{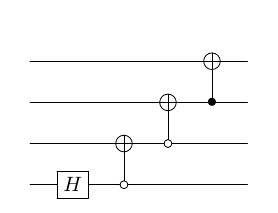}
        \caption{Circuit for the Unitary $G_4$ for $|1_L\rangle$. }
         \label{fig:uni4-1}
     \end{subfigure}
        \caption{Circuits for the Unitary matrices $U_k$s for $|1_L\rangle$. The hollow circles in the controlled-NOT are denoting the  controlled operation on the state $|0\rangle$ and the solid circles are denoting controlled on the  $|1\rangle$.  }
        \label{fig:Uk-1s}
\end{figure}
\begin{table}[t!]
    \centering
    \begin{tabular}{|c||c|c |}
    \hline
   \multirow{2}{12 em}{Recovery circuit on different devices}      
         & single qubit  & Two-qubit   \\
        & & \\  \hline
     Ideal Simulator    & 10          & 19\\
     Ideal Simulator ($|1_L\rangle$)   & 6         & 15\\
     IBM-Brisbane    & 64         & 47\\
      IBM-Brisbane ($|1_L\rangle$)    & 90         & 28\\
    \hline
    
    \end{tabular}
    \caption{Required number of resources to execute the error correction experiment on the hardware and ideal simulators. }
    \label{tab:gates_tab}
\end{table}
\begin{figure}[t!]
    \centering
    \includegraphics[width=1\columnwidth]{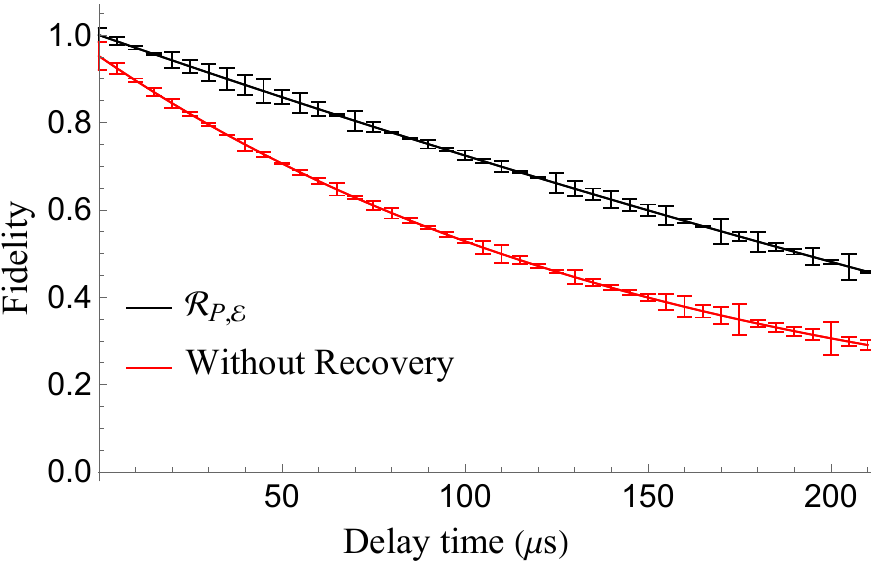}
    \caption{Performance of the syndrome based Petz map on the IBM-Torino (mocked backend). In this experiment we stick to the single cycle QEC only. The delay gates with time $t (\mu s)$ are appended in between the QEC and the encoding. For the bare qubit we apply the  We apply the same amount of delay. The delay time ranges between the $0 -230 \mu s $. On the $Y-$axis we plot the fidelity for the $|1\rangle$ state.    }
    \label{fig:petz_torr}
\end{figure}

Having optimized the two-qubit gate counts we can now aim for executing the recovery operation on the hardware. We choose two types of processors, one is the IBM-Brisbane, which is the Eagle processor and the other is the Heron processor, IBM -Torino. 
here in our experiment we take two ways to carry out the QEC through the Petz map. One is with the single cycle QEC and the other is with the two cycles QEC. On the Torino we executed one cycle QEC to generate the plot in Fig.\ref{fig:petz_torr} and on the Brisbane we have executed both the two cycle and one cycle QEC in generating the plots in Fig.\ref{fig:petz_marr}. For the lower amount of the delay on the Brisbane we have applied the one cycle of QEC and for the larger amount of delay ($t > 20$) we have considered two cycle of QEC. For the two cycles of QEC we applied the QEC at every $t/2 - \Delta t$ instant of time, where the $t$ is the time during which the noisy evolution (meaning without QEC) of the qubit takes place and $\Delta t$ is the time required for the recovery. On the Heron processor the native two-qubit gates are $CZ$ -- the controlled-$Z$ operations and \emph{gate-time} for these the $CZ$ gates are $64 \,n s$. Hence, in the Heron type processor the circuit execution time for our QEC (single cycle) circuit is around $2 \mu s$. However, in the Eagle processors the cross-resonance gate is the native two qubit gate, it is called $ECR$ gate. The gate time of a single $ECR$ is $660 \,ns$. Because of the larger gate time on Eagle processor the execution time of a single cycle QEC increases to $10 \,\mu s$.           
For more detail discussion on the multi cycle QEC set up we refer to the Appendix \ref{sec:multi_qec} and \ref{app:noise_sim}. With these strategies of multi and single round QEC we executed the $T_1$ experiment on the Eagle and the heron processors of IBM. We demonstrate the performance of the QEC with the syndrome Petz on the hardware in the Fig.\ref{fig:petz_marr}. The Figs.\ref{fig:petz_marr} and \ref{fig:petz_torr} clearly shows the improvement of the fidelity after the error correction. In order to measure the fidelity we can go with two directions, one is the state-tomography on the data qubits after the QEC. We do not follow this method since the multi-qubit state tomography is a time consuming process and thus cost more runtime on the hardware. We can perform a single qubit state tomography on the input qubit after the decoding (inverse of encoding) to estimate the fidelity, however this method is also runtime consuming (typically costs $\sim 12 s$ of runtime).  To save the runtime, we  bypass the need of tomography.  To measure the fidelity after the recovery we apply the inverse of the encoding just after the QEC and measure the ``input qubit" on the $Z-$ basis.
The complete method, starting from the encoding followed by the recovery and the inverse of encoding is exactly captured by the following operation
\begin{align}\label{eq:cw}
    \cW & \overset{\rm def}{=} U^{\dagger}_{\rm en} \circ\cR_{P,\cE}\circ \cA \circ U_{\rm en}.
\end{align}
The fidelity that we are intended to measure is the following 
\begin{align}
    F_s^2(|\psi\rangle) &= \langle \psi|\cW( |\psi\rangle\langle\psi|)|\psi\rangle.
\end{align}
Note that the fidelity $F_s^2$ serves as a cost function for finding the good code in \cite{ak_good_code} and the defining the optimality of the transpose channel \cite{hkn_pm2010}. The circuit in Fig.\ref{fig:full_rec} implements the operation $\cW$ in Eq.\eqref{eq:cw} and measuring the input qubit (the second qubit counted from the top) results in $F_s^2$. In our plots we denote the $F_s^2$ as ``Fidelity". In the Appendix \ref{app:fid_estimation} we shows just by measuring the input qubit we can estimate the fidelity $F^2_s$. To observe the improvement on the $T_1$ time we initiate the input qubit in $|1\rangle$ state and after operation $\cW$ we measure the input qubit to obtain the fidelity $F^2_s(|1\rangle)$. The $F_s^2(|1\rangle)$ denotes the population of the state $|1\rangle$ after the error correction and hence we can estimate the improved $T_1$ after the QEC just by examining the variation of $F_s^2(|1\rangle)$ against the delay time (t). We show the variation of $F_s^2(|1\rangle)$ against the delay time $t$ in the Figs. \ref{fig:petz_marr} and \ref{fig:petz_torr}. 

Making use of these data in for the $F_s^2(|1\rangle)$ we can estimate the lifetime (the $T_1$ time) of both the bare qubit and the error corrected qubit. To estimate the $T_1$ value before and the after the QEC we fit the data for the fidelity with the function $ f(t)= a + b \exp(-t /T)$.  After the fitting with the data for the bare qubit (without QEC ) we obtained the  qubit lifetime $T_{\rm bare} \sim 337 \mu s$ and $T_{\rm bare}\sim 155 \mu s$ for the IBM Brisbane and IBM Torino (fake-backend) respectively. The fitting with the data after the recovery results in $T \sim 676 \mu s$ on Brisbane and $T\sim 777 \, \mu s $ on Torino. For both the processors we kept the tolerance for the curve fitting around $\sim 10^{-8}$. Therefore we can see that the improvement in $T$ on IBM-Brisbane is about $2 \times T_{\rm bare}$ and on the IBM-Torino it is about  $5 \times T_{\rm bare}$. However, we note that the recent progresses on the QEC report a gain of more than two-times the bare qubit lifetime but they are with the perfect code and some of them are with bosonic codes \cite{girvin_break_even_2022}.     

\begin{figure}[t]
    \centering
    \includegraphics[width=1.0\columnwidth]{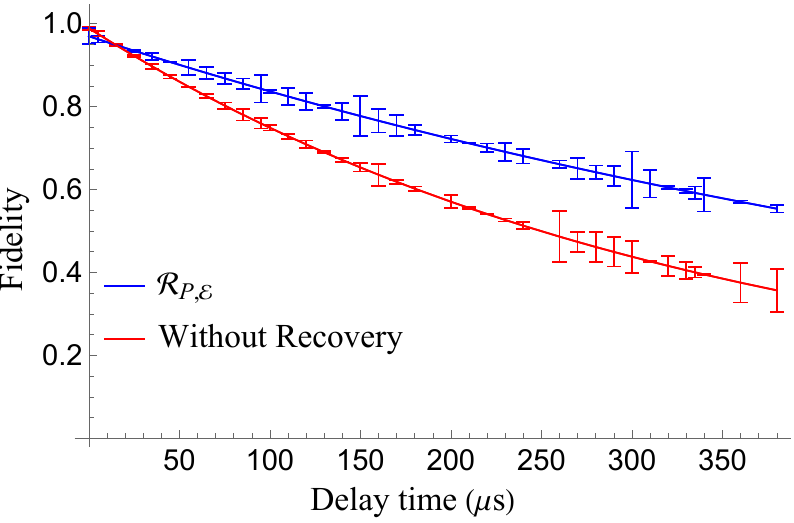}
    \caption{Fidelity \emph{vs.} the delay time in $\m s$ unit on IBM-Brisbane for single cycle QEC. On the Y-axis we plot the fidelity for the $|1\rangle$ state.}
    \label{fig:petz_marr}
\end{figure}

\section{Conclusion}

 Noise-adapted quantum error correction gives a framework to deal with the noise process by adapting it in designing the encoding, decoder, and recovery, which require fewer resources than the family picture framework of QEC through the general-purpose code. However, the framework of noise-adapted QEC generally belongs to the AQEC family. Here, we address the central challenge with the framework of AQEC -- the difficulties of detecting the syndromes as the noise subspaces overlap. We propose a subspace orthogonalisation algorithm to distinguish the subspaces, which gives a theoretical framework to measure the syndrome. Our framework makes the polar decomposition-based recovery feasible for those codes that do not satisfy the criteria posed by Leung \emph{et al.}. We also show that even for the [[4,1]] Leung code, the recovery $\cR_E$, which is a product of our subspaces orthogonalisation algorithm, performs better than the recovery proposed in \cite{leung}.

 We then show that the orthogonalisation algorithm makes some known decoders perform better, like the Petz map. This leads us to propose a new variant of the Petz recovery map, which we call the \emph{syndrome-based Petz recovery}. Although we do not know whether the syndrome-based Petz consistently outperforms the ordinary Petz,  we explicitly show that the syndrome-based Petz map performs optimally for the four-qubit Leung code and matches the performance of Fletcher SDP recovery. We also show that the syndrome-based Petz performs optimally with the optimal code proposed in \cite{liang_jiang_opt_code}. To this end, we note that Gram-Schmidt orthogonalisation (orthogonalisation for overlapping vectors) is helpful to define the search space of SDP \cite{sq_cats-2}. However, our algorithm differs from the usual Gram-Schmidt orthogonalisation procedure. 

 We then explore how well the noise-adapted recovery through our syndrome-based Petz map performs on the IBM hardware. Since amplitude-damping noise is the dominating noise process on the superconducting hardware,  we consider the four-qubit Leung code to correct it, as the code is adapted to the amplitude-damping noise. For the four-qubit code, we notice that the circuit complexity for the new variant of the Petz map is much lower than the previously proposed noise-adapted circuits \cite{biswas2024noise}. Moreover, this enables us to design a hardware-efficient circuit for the recovery. We carried out the circuit execution for the syndrome-based Petz with the [[4,1]]-Leung code on the IBM-Brisbane and IBM-Torino ( mock-backend ). We show that 
the improvement in the fidelity for the logical states is beyond the break-even point. The QEC circuit for the syndrome-based Petz shows a two times ($2\times$) improvement of the qubit lifetime on IBM-Brisbane and five times ($5\times$) on the mock-backend of IBM-Torino. 

Our framework provides a way to extract the syndrome for noise-adapted codes, or more generally, for approximate codes. Going forward, it will be interesting to explore how the framework of subspace discrimination can be applied to larger codes, such as the surface code with many qubits, or to bosonic codes—such as cat codes—where the erroneous subspaces overlap.

\section{Acknowledgment}
The authors thank Prof. B. M. Terhal and Prof. David P. DiVincenzo for insightful discussions. D.B thanks Vismay Joshi and Anubhab Rudra for the useful discussion on the QISKIT. This research was partly supported by a grant from the Mphasis F1 Foundation to the Centre for Quantum Information, Communication, and Computing (CQuICC). This work was also supported by a grant from the Defence Research and Development Organisation (DRDO).

\providecommand{\noopsort}[1]{}\providecommand{\singleletter}[1]{#1}%
%

\appendix

\section{Properties of the Kraus operators with orthogonal syndrome subspaces}\label{app:orth_channel}
Here, we prove a few useful properties of the new set of Kraus operators arising out of the orthogonalization algorithm in Sec.~\ref{sec:algorithm}. Recall that

\begin{lemma}\label{lem:UkU_l=0}
 The orthogonality $PE_{k}^{\dagger}E_{\ell}P =0$ implies that $PU_k^{\dagger}U_{\ell} P=0$.
 \end{lemma}
 \begin{proof}
     We consider the following 
     \begin{align}
      P U_k^{\dagger}U_\ell P = P\Tilde{M}_{kk}^{-1/2}P E_k^{\dagger}E_{\ell}P   \Tilde{M}_{\ell \ell}^{-1/2}P.
     \end{align}
     Here all the inverse operations are the pseudo-inverses and the operator $\Tilde{M}_{ii}= PE_{i}^{\dagger}E_iP$. Therefore using the orthogonality $PE_{k}^{\dagger}E_{\ell}P=0 $ for-all $k\neq \ell$ we have $PU_k^{\dagger}U_{\ell}P =0$. 
 \end{proof}

\begin{lemma}\label{lem:Ep<Ap}
    The operator $\cE(P)^{-1/2}- \cA(P)^{-1/2}$ is positive semi definite.
\end{lemma}

\begin{proof}
  Consider $|\phi\rangle$ be an arbitrary state in the support of $ Q_k$. Consider an other state $|\psi\rangle$ which we can write as
    \begin{align}
        |\psi\rangle & = \alpha|\phi\rangle +\beta|\phi_{\perp}\rangle.
    \end{align}\
    Therefore, taking the expectation value of the $\cE(P)$ with $|\psi\rangle$ we have following 
    \begin{align}
        \langle\psi| E_k PE_k^{\dagger} |\psi\rangle & = |\alpha|^2 \langle\phi| A_k PA_k^{\dagger} |\phi\rangle\\
       \sum_k  \langle\psi| E_k PE_k^{\dagger} |\psi\rangle  & \leq\sum_k  \langle\psi| A_k PA_k^{\dagger} |\psi\rangle \,\,{\rm since}\,\,|\alpha|^2\leq1\\
         \langle\psi|\cE(P)|\psi\rangle & \leq  \langle\phi|\cA(P)|\phi\rangle
    \end{align}
    Therefore, $\cA(P) \succeq \cE(P)$ on the support of the operators $Q_k$s. Hence $\cE(P)^{-1/2}\succeq \cA(P)^{-1/2}$ and also $\cE(P)^{-1/4}\succeq \cA(P)^{-1/4}$.
\end{proof}

    \begin{lemma}\label{lem:M>Mt}
        The modified QEC matrix $\tilde{M}_{kk}$ is lesser than $M_{kk}$.
    \end{lemma}
    \begin{proof}
    For the QEC matrix for the new set of Kraus operator we first note $\tilde{M}_{kl}= 0$ for $k\neq l$ and for the $k=l$ we have the following 
    \begin{align}
      \Tilde{M}_{kk}&=  PE_k^{\dagger}E_kP =P A_k^{\dagger}A_kP- PA_k^{\dagger}W_{k-1}A_kP-\nonumber\\&  PA_k^{\dagger}W_{k-1}A_kP + P A_k^{\dagger}W_{k-1}W_{k-1}A_kP.  
    \end{align}
    
    Here $W_{k-1}= \sum_{i=1}^{k-1} U_iP_iU_i^{\dagger}$. We note that $W_{k-1}W_{k-1}= \sum_{i,j=1}^{k-1} U_iP_iU_i^{\dagger}U_jP_iU_j^{\dagger}  $. Since $P_iU_i^{\dagger}U_jP_j=\delta_{ij}$, we the following form of the QEC matrix $\tilde{M}_{kk}$
    \begin{align}
         \Tilde{M}_{kk}&=  PE_k^{\dagger}E_lP =P A_k^{\dagger}A_kP- PA_k^{\dagger}W_{k-1}A_kP.
    \end{align}
    We also note that $\sum_{i=1}^{k-1}U_iP_iU_i^{\dagger} \succeq 0$ and $PA_k^{\dagger}(.)A_kP$ is also  positive semi definite. Therefore, 
    \begin{align}
        M_{kk}-\tilde{M}_{kk} \succeq 0,
    \end{align}
    where $M_{kk}= PA_k^{\dagger}A_kP$. Therefore, $\tilde{M}_{kk}^{-1/2}\succeq M^{-1/2}_{kk}$. 
\end{proof}

\section{Kraus operators of the syndrome-based Petz map for the $[4,1]$ code}\label{app:kraus_sbpm}
To derive a more concise form of the Kraus operator of the recovery operation $\cR_{P,\cE}$, we first note from the article \cite{liang_jiang_near_opt} that the $\cE(P)^{-1/2}$ has the following form 
\begin{align}\label{eq:ep_half}
\cE(P)^{-1/2} = \sum\limits_{\m,\n,k}(\tilde{M}^{-3/2})_{[\m,k],[\n,k]}E_k|\m\rangle\langle \n|E_k^{\dagger} 
\end{align}
Therefore the above expression for the $\cE(P)^{-1/2}$ we can derive the analytical expression for the kraus operators for the $\cR_{P,\cE}$ as follows 
\begin{align}
    R^{P,\cE}_l & = PE_l^{\dagger}\cE(P)^{-1/2}\\
    & = \sum\limits_{\m,\n, k,a} (\tilde{M}^{-3/2})_{[\m,k],[\n,k]}|a\rangle\langle a| E_l^{\dagger}E_k|\m\rangle\langle \n|E_{k}^{\dagger}\\
    & =  \sum\limits_{\m,\n,a} (\tilde{M}^{-3/2})_{[\m,l],[\n,l]} \tilde{M}_{[a,l],[\m,l]}|a\rangle\langle \n|E_{l}^{\dagger}\\
    & = \sum\limits_{\n,a} (\tilde{M}^{-1/2})_{[a,l],[\n,l]} |a\rangle\langle \n|E_{l}^{\dagger}
\end{align}
As noted in the article \cite{liang_jiang_near_opt}, here also we can rewrite the Kraus operator $R^{P,\cE}_k$ as follows 
\begin{align}
    R^{P,\cE}_k &= \sum\limits_{\m}|\m\rangle\langle \psi_{[\m k]}|,
\end{align}
where $|\psi_{[\m k]}\rangle= \sum\limits_{\n}(\Tilde{M}^{-1/2})^*_{[\m k],[\n k]}E_k|\n\rangle$. Note that $\{|\psi_{[\m k]}\rangle\}$ forms an orthonormal basis if $\tilde{M}_{[\m k],[\n k]}\propto \delta_{\m,\n}$. Following the Ref.\cite{cafaro_PhysRevA.89.022316} we note that, if $\{|\psi_{[\m k]}\rangle\}$  forms an orthonormal basis, then there exists an isometry $G_k$ such that $G_k |\psi_{[\m k]}\rangle= |\m\rangle$. Therefore, we can express the Kraus operators $R_{k}^{P,\cE}$ as 
\begin{align}
    R_{k}^{P,\cE}& = G_k \sum\limits_{\m,\n} |\psi_{[\m k]}\rangle \langle \psi_{[\m k]}|\\
    & = \sum\limits_{\m}|\m\rangle\langle \psi_{[\m k]}|.
\end{align}

For example, if we take the $[[4,1]]$ code in Eq.\eqref{eq:leung}, and the amplitude-damping noise, we note that $\tilde{M}_{[\m k],[\n k]}\propto \delta_{\m \n}$. Therefore, we the following 
\begin{align}
     |\psi_{[0_L, 1000]}\rangle &= \sum\limits_{\n =0,1}(\tilde{M}^{-1/2})_{[\n,1000],[0_L, 1000]} D_{1000}|\nu\rangle\\
     &=(\tilde{M}^{-1/2})_{[0_L,1000],[0_L, 1000]} D_{1000}|0_L\rangle.
\end{align}
Similarly, 
\begin{align}
     |\psi_{[1_L, 1000]}\rangle &= \sum\limits_{\n =0,1}(\tilde{M}^{-1/2})_{[\n,1000],[1_L, 1000]} D_{1000}|\nu\rangle\\
     &=(\tilde{M}^{-1/2})_{[1_L,1000],[1_L, 1000]} D_{1000}|0_L\rangle.
\end{align}
Now we have, 
\begin{align}
    \tilde{M}^{1/2}_{[0_L,1000],[0_L,1000]}& = \frac{1}{\sqrt{2}}\gamma^{1/2}(1-\gamma)^{3/2}\\
    \tilde{M}^{1/2}_{[1_L,1000],[1_L,1000]}& = \frac{1}{\sqrt{2}}\gamma^{1/2}(1-\gamma)^{1/2}.
\end{align}
We also note that the action of the error $D_{1000}$ takes the logical vectors $\{|0_L\rangle,|1_L\rangle\}$ of the four-qubit code in Eq.\eqref{eq:leung} to the following vectors 
\begin{align}
   D_{1000} |0_L\rangle & = \frac{\gamma^{1/2}(1-\gamma)^{3/2}}{\sqrt{2}}|0111\rangle\\
    D_{1000} |1_L\rangle & = \frac{\gamma^{1/2}(1-\gamma)^{1/2}}{\sqrt{2}}|0100\rangle.
\end{align}
Therefore the Kraus operator of the syndrome-based Petz map correcting the error $D_{1000}$ is the following 
\begin{align}
    R_1= |0_L\rangle\langle0111|+ |1_L\rangle\langle0100|,
\end{align}
which matches exactly with the listed operators in the Table \ref{tab:rec_tab}, where we got the operators from our numerics.

Furthermore, we note by comparing the description of the Kraus operators for the map $\cR_{P,\cA}$ from \cite{liang_jiang_near_opt} with the $R_l$ that we have just derived we note that the only difference between them is the QEC matrix $M$ and $\tilde{M}$. While the $\tilde{M} \propto \delta_{kl}$ the matrix $M$ is not proportional to $\delta_{kl}$ in general. Moreover, we prove that $\tilde{M} \leq M$ in the Appendix \ref{app:orth_channel}.
\section{Construction of the new set of Kraus operators for the amplitude-damping noise for any four-qubit code }\label{app:new_kraus_cons}
We consider $P$ to be the projector onto any of the four-qubit code spaces $\cC$ considered in Sec.~\ref{sec:ex}. In our work, we consider three four-qubit noise-adapted codes that correct amplitude-damping noise only. Among these codes, two are the optimal codes obtained from two different numerical searches. The set of Kraus operators for the four-qubit Hilbert space is the following:
\begin{align}\label{eq:amp_damp_kraus}
    \{D_{0000},D_{0001},D_{0010},D_{0100},D_{1000},D_{1001},D_{0110},D_{0101}\nonumber\\,D_{1010}D_{0011},D_{1100},D_{0111},D_{1011},D_{1101},D_{1110},D_{1111}\}.
\end{align}
To construct the new set of Kraus operators $E_k$, we first identify $E_1 P = D_{0000} P$. Therefore, $P_1$, which is the projector onto the support of $P E_1^{\dagger} E_1 P$, is simply $P$, as $E_1$ does not annihilate any of the code vectors. Next, we consider the Kraus operators $E^{(1)} \equiv {E_2, E_3, E_4, E_5}$, which are the single-weight error operators ${D_{0001}, D_{0010}, D_{0100}, D_{1000}}$, and then we choose the two-weight damping errors $E^{(2)} \equiv {D_{1001}, D_{0110}, D_{0101}, D_{1010}}$ as ${E_5, E_6, E_7, E_8}$. 

For the single-weight damping errors, the $P_i$’s, which are the projectors onto the support of $P E_i^{\dagger} E_i P$, are simply $P$, the projector onto $\cC$. This is because none of the operators in $E^{(1)}$ annihilates any code vectors. However, the operators in $E^{(2)}$ annihilate $|1_L\rangle$ in Eq.~\eqref{eq:leung}. Therefore, the support of $P E_i^{\dagger} E_i P$ for $E_i \in E^{(2)}$ contains only one logical vector.

Now we consider the errors $D_{0011}$ and $D_{1100}$. Suppose we choose the error $D_{1100}$ to construct $E_9$ with the code in Eq.~\eqref{eq:leung}:
\begin{align}
    E_{9}P&= (I - \sum\limits_{i =0}^{8}U_iP_iU_i^{\dagger})D_{1100}P,
\end{align}
we see that $P_9$ is simply $P$, since $E_9$ does not annihilate any logical vectors and $\sum\limits_{i=1}^{10} U_i P_i U_i^{\dagger} = I$. Therefore, one can see that the following operator
\begin{align}
    E_{10}P&= (I - \sum\limits_{i =0}^{9}U_iP_iU_i^{\dagger})D_{0011}P,
\end{align} 
is a null-matrix. For the optimal code in Eq.~\eqref{eq:liang_code}, the scenario is different from the four-qubit code in Eq.~\eqref{eq:leung}, as we will discuss in the subsequent sections. We now proceed to construct the syndrome-based Petz map for the codes and examine whether the syndrome-based Petz ($\cR_{P,\cE}$) outperforms the Petz map ($\cR_{P,\cA}$) or not.

\subsection{[[4,1]]-Leung code}
Having constructed the new set of Kraus operators with the [[4,1]]-Leung code in Eq.~\eqref{eq:leung}, we now proceed to construct the syndrome-based Petz map. We note that the orthogonalisation algorithm results in ten Kraus operators ${E_i P}$, while there are sixteen error operators $D_{ijkl} P$ in the actual noise process. With these ten operators, we construct the Kraus operators for $\cR_{P,\cE}$, which are listed in Table~\ref{tab:rec_tab}.

We note two crucial differences in the Kraus operators of $\cR_{P,\cE}$ compared to those of $\cR_{P,\cA}$. The first difference is that none of the operators ${R_2 - R_9}$ depends on the damping strength $\gamma$. These operators, i.e., $R_2 - R_9$, correct the weight-one and weight-two errors, while the corresponding Petz recovery operators depend on $\gamma$. We also note that for this code, $A_k P = E_k P$ for $1 \leq k \leq 9$, but $E_{10} P \neq A_{10} P = D_{0011} P$.
\subsection{[[4,1]]-Optimal code in Eq.\eqref{eq:liang_code}}
The numerically obtained code in Eq.~\eqref{eq:liang_code} performs optimally against the amplitude-damping noise under the optimal recovery, which is the result of a numerical search based on semi-definite programming (SDP), as reported in \cite{liang_jiang_opt_code}. For the code in Eq.~\eqref{eq:liang_code}, we first identify $E_0 P = D_{0000} P$ and construct all the operators ${E_1, E_2, E_3, E_4}$ for the weight-one errors.

For the weight-two errors, we consider two different flows of the orthogonalisation algorithm. In the first flow, we construct the new operators as follows:
\begin{align}
    E_5P= (I -\sum\limits_{i = 0}^{4}U_iP_iU_i^{\dagger})D_{1001}P\\
    E_6P= (I -\sum\limits_{i = 0}^{5}U_iP_iU_i^{\dagger})D_{0110}P\\
    E_7P= (I -\sum\limits_{i = 0}^{6}U_iP_iU_i^{\dagger})D_{0011}P\\
    E_8P= (I -\sum\limits_{i = 0}^{7}U_iP_iU_i^{\dagger})D_{1100}P\\
    E_9P= (I -\sum\limits_{i = 0}^{8}U_iP_iU_i^{\dagger})D_{1010}P\\
    E_{10}P= (I -\sum\limits_{i = 0}^{9}U_iP_iU_i^{\dagger})D_{0101}P,
\end{align}
and we observe that $E_{10}$ is a null matrix. We note that the support of $P E_9^{\dagger} E_9 P$ contains only one vector, $|0_L\rangle$.

In another flow of the algorithm, we consider the following:
\begin{align}
    E_5P= (I -\sum\limits_{i = 0}^{4}U_iP_iU_i^{\dagger})D_{1001}P\\
    E_6P= (I -\sum\limits_{i = 0}^{5}U_iP_iU_i^{\dagger})D_{0110}P\\
    E_7P= (I -\sum\limits_{i = 0}^{6}U_iP_iU_i^{\dagger})D_{0101}P\\
    E_8P= (I -\sum\limits_{i = 0}^{7}U_iP_iU_i^{\dagger})D_{1010}P\\
    E_9P= (I -\sum\limits_{i = 0}^{8}U_iP_iU_i^{\dagger})D_{1100}P\\
    E_{10}P= (I -\sum\limits_{i = 0}^{9}U_iP_iU_i^{\dagger})D_{0011}P,
\end{align}
and again we note that $E_{10} P$ is a null matrix. However, in this case, the support of $P E_7^{\dagger} E_7 P$ and $P E_8^{\dagger} E_8 P$ contains two vectors. However, if we consider the projectors $(P_7, P_8) = |0_L\rangle \langle 0_L|$ and then proceed with the construction of $R_{P,\cE}$, we achieve the optimal performance with this code.

\subsection{Structured search code}
For the structured search code in \cite{ak_good_code}, we considered two ways to choose the sequence of the kraus operators in order to construct the $E_k$s. Case 1: Here we choose an arbitrary sequence of the Kraus operators and in Case 2: we choose the $E_1P=A_1P$ and the $E_{i}P$ and $E_{i+1}P$ following the inequality  $PA_{1}^{\dagger}A_{i}P > PA_{1}^{\dagger}A_{i+1}P$. We see that the entanglement fidelity from the case 1 is the $1- 1.53\gamma^2 +\cO(\gamma^3)$   as noted in the Table \ref{tab:wcf_table} and in the Case 2, results in the  fidelity $1- 1.45\gamma^2 +\cO(\gamma^3)$. Note that in all this cases (Case 1 and Case 2), we have fitted the numerical data with a polynomial $1+ b \,x^2+ c \,x^3 $. 

\section{Cycles of syndrome based Petz map}\label{sec:multi_qec}
In the implementation section, we have considered two cases for the hardware implementation of the syndrome-based Petz map. One is the single-cycle QEC on the IBM-Torino, and the other is the two-cycle QEC on the Eagle processor of IBM Brisbane. The strategy for the single-cycle QEC is very straightforward: we prepare the encoded state, and then, after the noise, we apply the syndrome followed by the recovery.

However, for the two-cycle QEC, our strategy is more intricate. It is intricate because, after the noise acts on the encoded state, we do not apply two consecutive rounds of QEC. Rather, we consider that if the noise acts for $t$ units of time, we apply the recovery every $(t/2)$ units of time; i.e., we apply the following sequence of operations: 
\begin{align}
    \cR_{P,\cE^{t/2}} \circ \cA^{t/2} \circ  \cR_{P,\cE^{t/2}} \circ \cA^{t/2},
\end{align}
where $\cE^{t/2}$ is the channel resulting from the orthogonalisation algorithm acting on the $\cA^{t/2}$. For the $N-$ cycles of QEC, we consider the following sequence operations:  
\begin{align}\label{eq:n_cycle_petz}
    \underbrace{\cR_{P,\cE^{t/N}} \circ \cA^{t/N}}_{N^{th}- \mbox{cycle}} \circ  \underbrace{\cR_{P,\cE^{t/N}} \circ \cA^{t/N}}_{N-1^{th} \mbox{cycle}}\cdots \\\underbrace{\cR_{P,\cE^{t/N}} \circ \cA^{t/N}}_{\mbox{first cycle }}(|1_L\rangle\langle1_L|).
\end{align}

Now we study the behaviour of the multi-cycle QEC through the sequence of the syndrome-based Petz map in Eq.\eqref{eq:n_cycle_petz} on an ideal and noisy simulator. 

\subsection{Ideal simulation}\label{sec:ideal_multi_QEC}

We execute the sequence in Eq.\eqref{eq:n_cycle_petz} on the ideal simulator, where the qubits face only amplitude-damping noise and the gates are ideal. We observe that as the number of cycles, i.e., the $N$ in Eq.\eqref{eq:n_cycle_petz} of the QEC, increases, the fidelities also increase. For example, if the qubits are subjected to amplitude-damping noise of strength $\gamma=0.1$, the increment of the fidelity follows the trend shown in Fig.\ref{fig:fid_petz_cycle}.

We then carry out the simulation of QEC with multicycle syndrome-based Petz for different choices of the damping strength $\gamma$ and observe that when the single-cycle syndrome-based Petz gives the fidelity $F^2 (\cR_{P,\cE}\circ\cA (|1\rangle\langle1|))= 1- \gamma^2$, the multicycle syndrome-based Petz with $N=5$ results in $F^2 (\cR_{P,\cE}\circ\cA (|1\rangle\langle1|))= 1- 0.2\gamma^2$. We estimate the analytical expression through a polynomial fit on the numerically generated data in Fig.\ref{fig:multi-cycle_QEC}. However, the multicycle syndrome-based Petz behaves differently on the noisy simulation as we will see in the following section.

\begin{figure}
    \centering
    \includegraphics[width=1\linewidth]{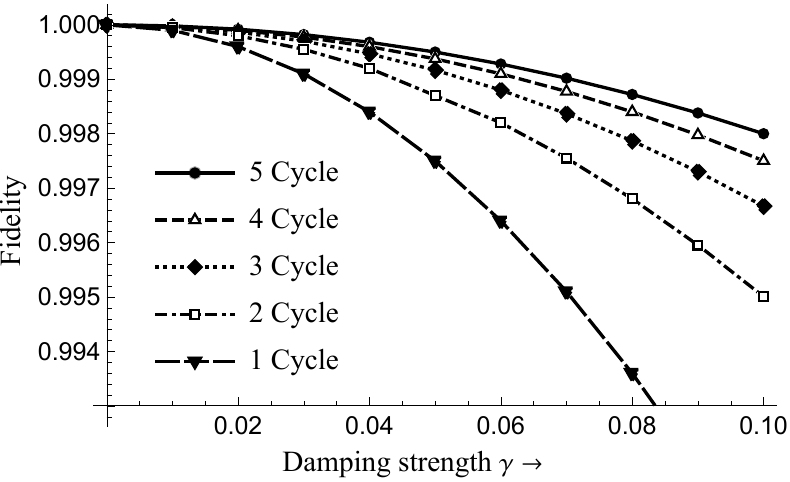}
    \caption{Ideal Simulation of Multicycle QEC through the Petz Map. On the Y-axis, we plotted the \textit{fidelity} between the initial state $|1_L\rangle$ and the recovered state, i.e., 
$\langle 1_L | \mathcal{R}_{P,\mathcal{E}} \circ \mathcal{A} (|1_L\rangle \langle 1_L|) | 1_L \rangle$,
as obtained in the ideal simulation.}
    \label{fig:multi-cycle_QEC}
\end{figure}
\begin{figure}
    \centering
    \includegraphics[width= 1\columnwidth]{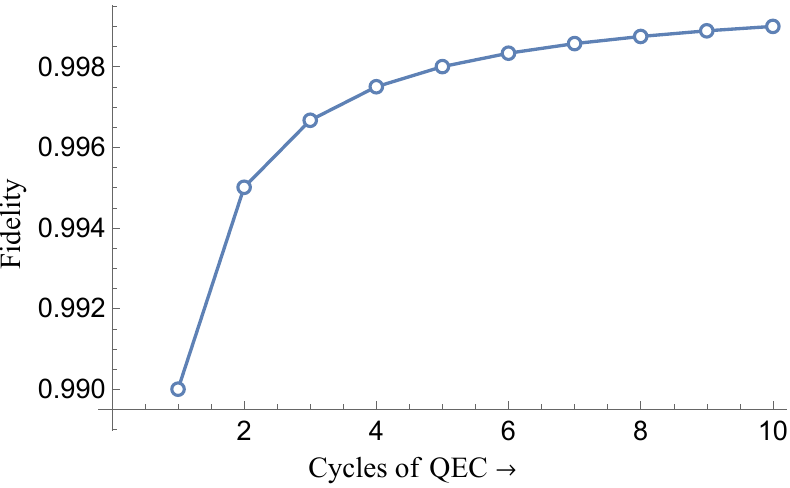}
    \caption{Performance of the syndrome-based Petz map with different cycles for a fixed value of the damping parameter $\gamma=0.1$.}
    \label{fig:fid_petz_cycle}
\end{figure}

\subsection{Noisy simulations}\label{app:noise_sim}
In the previous section (Sec. \ref{sec:ideal_multi_QEC}), we discuss the possible outcomes and the behavior of the multicycle QEC through the map on an ideal simulator. In this section, we aim to study the same on noisy simulators, where we customize the noise of the gates and qubits. On the customized backend, we consider two kinds of setups. In one setup, we keep the noise on the qubit as amplitude damping by fixing \(T_1 = T_2/2\), while in the other setup, we randomly choose the values of \(T_1\) and \(T_2\). To apply the noise on the qubit, we let the qubit evolve in the noisy environment by applying a delay of time \(t\). Since we are performing the \(T_1\) experiment, the considered initial state is \(|1\rangle\), and for any values of \(T_1\) and \(T_2\), the population of the state \(|1\rangle\) follows an exponential decay as \(\exp(-t/T_1)\). The gate error can range from \(10^{-1}\) to \(10^{-3}\) in both setups. The reason for this choice is to study the robustness of our protocols against gate errors, and this simulation also helps us gain insight into how robust our circuit is against hardware noise. However, we stress that our noise model is not time-dependent, whereas the hardware noise is, as pointed out in \cite{gupta2025expedited}. Therefore, the results may not fully match those obtained from noisy simulations.

In carrying out the multicycle QEC, we first note that the recovery procedure takes a finite amount of time. If we customize our noisy backend with gate times and gate errors similar to those of the Heron processors, the recovery takes \(2~\mu s\), while on a backend with gates similar to the Eagles, the syndrome-based Petz recovery takes \(5~\mu s\). Taking into account the recovery time, we consider the following sequence to execute the multicycle QEC:
  
\begin{figure}[t]
    \centering
    \includegraphics[width=1.0\columnwidth]{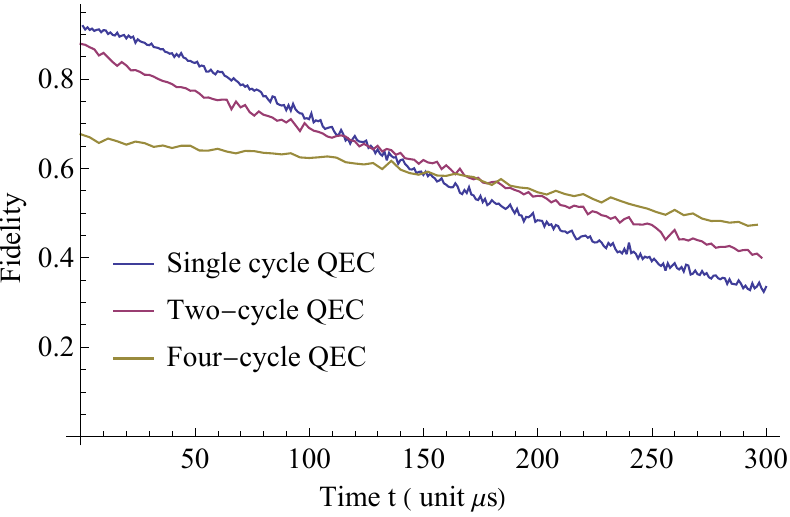}
    \caption{The figure shows the outcome of the $T_1$ experiment on a customized backend with different cycles of the syndrome-based Petz map. Here, the qubits are undergoing decoherence from the amplitude-damping noise. }
    \label{fig:fake_ecr}
\end{figure}
\begin{figure}[t]
    \centering
    \includegraphics[width=1\columnwidth]{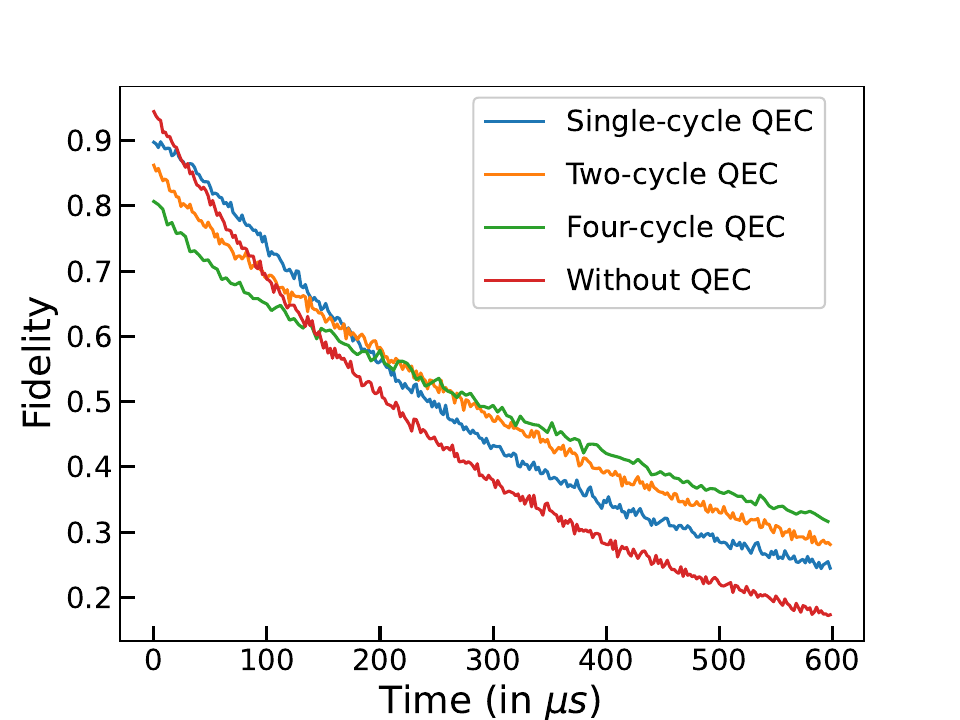}
    \caption{Noisy simulation of the multicycle syndrome-based Petz map, where the noises on each qubit are different and need not be amplitude-damping. }
    \label{fig:noisy_sim}
\end{figure}
\begin{align}
    R_{P,\cE^{t/N-\Delta t}}\circ\cA^{t/N - \Delta t} \circ  R_{P,\cE^{t/N-\Delta t}}\circ\cA^{t/N - \Delta t} \circ \cdots\nonumber \\
    \circ  R_{P,\cE^{t/N-\Delta t}}\circ\cA^{t/N - \Delta t},
\end{align}
where the $\Delta t $ is the time taken by $\cR_{P,\cE^{t/N-\Delta t}}$.
To study the recovery behavior through the syndrome-based Petz, we consider both the single-cycle and multicycle QEC setups as discussed in the previous section. In contrast to the ideal simulation, we observe that the performance of the multicycle recovery deteriorates with an increase in the number of cycles \(N\). However, this deterioration occurs for smaller delay time values \(t\). We also observe that the multicycle QEC performs better than the single-cycle QEC for larger delay times. The different performances of the multicycle QEC are shown in Fig.~\ref{fig:fake_ecr}.  

In generating the data for Fig.~\ref{fig:fake_ecr}, we consider the maximum \(ECR\) error to range between \(\sim 0.01\) and \(\sim 0.006\), with \(T_1 = 155~\mu s\) and \(T_2 = 2T_1\). Moreover, we assume that \(T_1 \,\&\, T_2\) are the same for all qubits. When we consider the overall performance of the multicycle QEC, we observe that the fidelity decays much more slowly than that of the single-cycle QEC.  

We fit the data for the multicycle QEC using the function \(f(t) = a + b \exp[-t^c /T^c]\). For the four-cycle QEC, the improved lifetime is \(T_1 = 2.89~{\rm ms}\), for the two-cycle QEC it is \(T_1 = 2.5~{\rm ms}\), and for the single-cycle QEC, it is \(T_1 = 2.01~{\rm ms}\).
   
\begin{table}[t]
    \centering
    \begin{tabular}{|c|c|c|}
    \hline
         & Max. $\And$ Min. \textsc{ECR} errors & $T_1$ and $T_2$ in $\m s$ \\
         \hline
     Fig.\ref{fig:fake_ecr}    & 0.1 $\And$ 0.001 & $T_1=155$, $T_2= 310$ \\
      Fig.\ref{fig:noisy_sim}    & 0.009 $\And$ 0.005 & $T_1=297$, $T_2= 85$ \\
      \hline
    \end{tabular}
    \caption{Table for the gate noise and the $T_1$, $T_2$ time for the input qubit. For the Fig.\ref{fig:fake_ecr}, we have considered purely amplitude-damping noise for the qubits, but for the Fig.\ref{fig:noisy_sim}, the noise model is from the fake-backend, similar to the Brisbane. For both the figures (\ref{fig:noisy_sim} and \ref{fig:fake_ecr} the errors on the  single qubit gates are of the order of $\sim 10^{-4}$.)}
    \label{tab:fake_table}
\end{table}

On the other hand, we consider the fake backend that mimics the noise of IBM Brisbane and choose the qubit parameters specified in Table~\ref{tab:fake_table}. In this backend, the average measurement error is \(4.44\%\). The \(ECR\) error typically ranges between \(9 \times 10^{-3}\) and \(5 \times 10^{-3}\). We generated a similar plot for the multicycle QEC with this range of gate noise. From the data in Fig.~\ref{fig:noisy_sim}, we observe that the lifetime of the qubit does not improve as much as in the case with purely damping noise. After fitting the numerical data with the function \(f(t) = a + b \exp[-t /T]\), we obtain an improved lifetime of \(T_1 = 562~\mu s\) with four QEC cycles, while with two cycles of QEC, we have an improved lifetime of \(T_1 = 490~\mu s\). A single cycle of QEC improves the lifetime to \(T_1 = 384~\mu s\).  

This result serves as guidance for implementing QEC on hardware. The hardware qubit and gate noise specifications are shown in Table~\ref{tab:tab_fidelity_1}. We see that in the real backend of Brisbane, the maximum and minimum \textsc{ECR} noise values are \(8 \times 10^{-3}\) and \(3 \times 10^{-3}\), respectively, which are lower than those of the mock backend. Hence, a better performance on the actual hardware is expected. Thus, we observe approximately a twofold gain in lifetime on the Eagle processors of IBM.


\section{Fidelity estimation}\label{app:fid_estimation}
Here we outline the procedure for estimating the fidelity using the circuit shown in Fig.~\ref{fig:full_rec}. As mentioned earlier, adding the decoding circuit (the inverse of the encoding circuit) is necessary to keep the runtime below \(10~\mathrm{s}\) on the IBM hardware. Typically, the runtime for the \(T_1\)-experiment circuit in Fig.~\ref{fig:full_rec} is around \(2~\mathrm{s}\) when measuring the input qubit (the second qubit in the circuit). If we perform single-qubit state tomography on the input qubit, the runtime increases to \(10~\mathrm{s}\). Without the decoding circuit, performing four-qubit state tomography directly on the data qubits immediately after recovery would increase the runtime to approximately four minutes.


 \begin{lemma}
 
Let the input state be $|\psi\rangle = |m\rangle$ with $m \in \{0,1\}$. The probability of obtaining the outcome $m$ upon measuring the second qubit in the circuit shown in Fig.~\ref{fig:full_rec} in the $Z$-basis is 
\begin{align}
    {\rm Prob}(|m\rangle) = \langle m| \tilde{\cR} \circ \cA (|m_L\rangle\langle m_L|) |m\rangle.
\end{align}
Moreover, ${\rm Prob}(|m\rangle)$ coincides with the fidelity between the logical state $|m_L\rangle$ and the recovered state $\tilde{\cR} \circ \cA (|m_L\rangle\langle m_L|)$.
\end{lemma}

\begin{proof}
To proceed with the proof of the claim of the Lemma, we first note that before the measurement on the input qubit, the circuit operation is equivalent to the following operation:
\begin{align}
 \cW &= U_{\rm en}^{\dagger} [\tilde{\cR} \circ \cA (U_{\rm en} |0m00\rangle\langle 0m00| U_{\rm en}^{\dagger})] U_{\rm en}.
\end{align}
We first prove the claim of the Lemma for $m=1$. Before executing the recovery $\tilde{\cR}$, we note that the noisy state $\cA(|1_L\rangle\langle 1_L|)$ is given by
\begin{align}
   \cA(|1_L\rangle\langle 1_L|) & = (1-\gamma)^2 |1_L\rangle\langle 1_L| + \frac{\gamma(1-\gamma)}{2} \sum\limits_{i=1}^{4} |\psi_i\rangle\langle \psi_i| \nonumber \\
   & \qquad\qquad + 2\gamma^2 |0000\rangle\langle 0000|,
\end{align}
where $|\psi_i\rangle \in \{|0100\rangle, |1000\rangle, |0010\rangle\}$, and these states are the images of the single-qubit damping error operators $\{D_{1000}, D_{0100}, D_{0010}, D_{0001}\}$. For more details, see \cite{leung}. We note that the two-qubit damping errors $\{D_{1001}, D_{0110}, D_{1010}, D_{0101}\}$, the three-qubit errors $\{D_{1110}, D_{0111}, D_{1011}, D_{1101}\}$, and the four-qubit damping error $D_{1111}$ annihilate the state $|1_L\rangle$.  

Now, we note that the recovery operators from Table~\ref{tab:rec_tab} and the error operators $A_i \in \{D_{0000}, D_{0001}, D_{0010}, D_{0100}, D_{1000}\}$ satisfy the following identity:
\begin{align}
   \langle 1_L| R_k N_\ell |1_L\rangle = 0 \quad \forall k \neq \ell.
\end{align}

\begin{table}[t!]
    \centering
    \begin{tabular}{c|c}
    \hline
      Input state & Output of $U_{\rm en}$ \\
      \hline
      $|0100\rangle$ & $|1_L\rangle$ \\
      $|0101\rangle$ & $X_1|1_L\rangle$ \\
      $|0110\rangle$ & $X_4|1_L\rangle$ \\
      $|0111\rangle$ & $|\tilde{1}_L\rangle$ \\
      $|1100\rangle$ & $X_4|\tilde{1}_L\rangle$ \\
      $|1101\rangle$ & $X_1X_4|1_L\rangle$ \\
      $|1110\rangle$ & $X_1|\tilde{1}_L\rangle$ \\
      $|1111\rangle$ & $X_1X_4|\tilde{1}_L\rangle$ \\
      \hline
    \end{tabular} \hspace{0.1cm} 
    \begin{tabular}{c|c}
    \hline
      Input state & Output of $U_{\rm en}$ \\
      \hline
      $|0000\rangle$ & $|0_L\rangle$ \\
      $|0001\rangle$ & $X_1|0_L\rangle$ \\
      $|0010\rangle$ & $X_4|0_L\rangle$ \\
      $|0011\rangle$ & $|\tilde{0}_L\rangle$ \\
      $|1000\rangle$ & $X_4|\tilde{0}_L\rangle$ \\
      $|1001\rangle$ & $X_1X_4|0_L\rangle$ \\
      $|1010\rangle$ & $X_1|\tilde{0}_L\rangle$ \\
      $|1011\rangle$ & $X_1X_4|\tilde{0}_L\rangle$ \\
      \hline
    \end{tabular}
    \caption{Truth table for $U_{\rm en}$. Here, $|\tilde{1}_L\rangle = \frac{1}{\sqrt{2}}(|1100\rangle - |0011\rangle)$, $|\tilde{0}_L\rangle = \frac{1}{\sqrt{2}}(|0000\rangle - |1111\rangle)$, $X_1 = XIII$ and $X_4 = IIIX$.}
    \label{tab:truth_table}
\end{table}

In our circuit in Fig.~\ref{fig:full_rec}, we implement only the recovery $\tilde{\cR}(.) = \sum\limits_{i=2}^{5} R_i (.) R_i^{\dagger}$. We do not include $R_0$ here, since it is an identity operation on the state $|1_L\rangle$ (see Table~\ref{tab:rec_tab}). Therefore, the recovered state is
\begin{align}\label{eq:rec_state}
   \tilde{\cR} \circ \cA(|1_L\rangle\langle 1_L|) &= \big[(1-\gamma)^2 + 4 \frac{\gamma(1-\gamma)}{2}\big] |1_L\rangle\langle 1_L| \nonumber \\
   & \qquad\qquad + 2\gamma^2 |0000\rangle\langle 0000| \\
   &= F^2_{|1_L\rangle} |1_L\rangle\langle 1_L| + 2\gamma^2 |0000\rangle\langle 0000|,
\end{align}
where $F^2_{|1_L\rangle} = F^2(\tilde{\cR} \circ \cA(|1_L\rangle), |1_L\rangle)$.  

From Table~\ref{tab:truth_table}, the encoding unitary $U_{\rm en}$ is
\begin{align}
    U_{\rm en} &= |1_L\rangle\langle 0100| + \sum_j |\phi^1_j\rangle\langle \chi^1_j| + |0_L\rangle\langle 0000| + \sum_j |\phi^0_j\rangle\langle \chi^0_j|,
\end{align}
where $\chi^1_j \in \{0101,0110,0111,1100,1101,1110,1111\}$, $\chi^0_j \in \{0001,0010,0011,1000,1001,1010,1011\}$, and the states $\{|\phi^0_j\rangle, |\phi^1_j\rangle\}$ are the outputs of the states $\{|j_0\rangle, |j_1\rangle\}$, as shown in Table~\ref{tab:truth_table}. We also note that $\langle m_L|\phi_j^m\rangle = 0$ for all $j$ and $m \in \{0,1\}$.  

Therefore, after applying the inverse of the encoding operation, we have
\begin{align}
    U_{\rm en}^{\dagger} (\tilde{\cR} \circ \cA(|1_L\rangle\langle 1_L|)) U_{\rm en} &= F^2_{|1_L\rangle} |0100\rangle\langle 0100| \nonumber \\
    & \qquad + \gamma^2 |0000\rangle\langle 0000|.
\end{align}
Tracing out all qubits except the input qubit, we obtain
\begin{align}
    \rho = \begin{pmatrix}
        \gamma^2 & 0 \\ 0 & F^2_{|1_L\rangle}
    \end{pmatrix}.
\end{align}
Therefore, ${\rm Prob}(|1\rangle) = F^2_{|1_L\rangle}$, which is the fidelity between $|1_L\rangle$ and the recovered state $\cR \circ \cA(|1_L\rangle\langle 1_L|)$. Upon simplifying Eq.~\eqref{eq:rec_state}, we have $F^2_{|1_L\rangle} = 1 - \gamma^2$.  

Similarly, for $m=0$, we have
\begin{align}
   \tilde{\cR} \circ \cA(|0_L\rangle\langle 0_L|) &= F^2_{|0_L\rangle} |0_L\rangle\langle 0_L| + 2 \gamma^3 (1-\gamma) |1_L\rangle\langle 1_L|,
\end{align}
where $F^2_{|0_L\rangle} = |\alpha|^2 (1 + \frac{\gamma^2}{2}) + |\beta|^2 (1-\gamma)^2 + \sqrt{2}(1-\gamma) \mathrm{Re}(\alpha \beta^*)$. After applying the inverse encoding, we obtain
\begin{align}
    U_{\rm en}^{\dagger} \tilde{\cR} \circ \cA(|0_L\rangle\langle 0_L|) U_{\rm en} &= F^2_{|0_L\rangle} |0000\rangle\langle 0000| \nonumber \\
    & \qquad + 2 \gamma^3 (1-\gamma) |0100\rangle\langle 0100|.
\end{align}
Tracing out all qubits except the second qubit, we have
\begin{align}
    \rho = \begin{pmatrix}
        F^2_{|0_L\rangle} & 0 \\
        0 & 2 \gamma^3 (1-\gamma)
    \end{pmatrix}.
\end{align}
\end{proof}

\section{Universal recovery based on polar decomposition}\label{sec:recovery}

We first note that the approximate QEC aided through the  Leung recovery requires the syndrome subspaces should be orthogonal to each other (the QEC condition in  Eq. (29) in Ref. \cite{leung}). Therefore, for a generic AQEC code Leung recovery is applicable since the syndrome subspaces may not be not orthogonal. Our subspace-discrimination algorithm enables the application of the Leung recovery for an arbitrary AQEC code.      

In performing the recovery operation, we first note that we are detecting those sub-spaces, which are the images of the new set of Kraus operators. The detection is possible since the $\{E_k\}$s maps the code space to orthogonal subspaces.  Therefore, in the recovery operation, we recover the code space from the image space of the operators $\{E_k\}$ through the unitary operators $U_k$s approximately. The recovery is still approximate since the deformation of the code space is still there even after the orthogonalisation.  Therefore,  like the Leung recovery, we define our recovery operation $\cR_{E} \equiv \{R_k, P_E\}$ as 
\begin{align}\label{eq:rec_new_1}
    R_k &= PU_k^{\dagger} = U_{k}^{\dagger}\Pi_{k},
\end{align}
where $\Pi_k= U_kPU_{k}^{\dagger}$ is the projector onto orthogonal syndrome subspaces.  We defined earlier that the $supp$ for projectors $P_i$ contains in the $supp$ of $P$. However it is necessary that there should exists sufficient numbers of projectors for which $supp (P_i) = supp(P) $. Otherwise the recovery operation would not able to correct the errors. For now we proceed with the recovery $\cR_E$. We denote the operator $PE_k^{\dagger}E_lP$ as $\tilde{M}_{kl}$. Therefore, the recovery operators can be re-written as 
\begin{align}\label{eq:rec_1}
    R_k&= \Tilde{M}^{-1/2}_{kk}P_k E_k ^{\dagger}.
\end{align}
It is quite straightforward to note that the recovery operators $\{R_k\}$ from a trace non-increasing channel as  
\begin{align}
    \sum\limits_{k} R_k^{\dagger}R_k &= \sum\limits_{k}U_{k}PU_k^{\dagger} \leq I.
\end{align}
 $\{U_kPU^{\dagger}_k=\Pi_k\}$ is a  set of mutually orthogonal projector because $PU^{\dagger}_{k}U_{\ell}P=0$. Since $\Pi_k$s are projecting to the mutually orthogonal subspaces, we can distinguish the subspaces. To implement the recovery through $\cR_E$ we first need to measure $\{\Pi_k\}$s to know which error has occurred. Then we apply the recovery through the unitary $U_k$. We show a syndrome detection circuit for the recovery $\cR_E$ in the Appendix \ref{app:uni_syn}.  However we can see in the Sec.\ref{sec:ex} the recovery by $\cR_{E}$ is not optimal or not even near optimal as it does not outperform the well known near optimal recovery -- \emph{The Petz map}. In achieving the optimality or at-least to have better performance than the ordinary Petz map, we consider a new version of the Petz map in the following section, which we call the \emph{syndrome-based Petz map}.

Now, we will provide the theoretical lower bound on the entanglement fidelity under our recovery protocol based on the polar decomposition through the following Lemma.

\begin{lemma}
    The worst-case fidelity under the recovery $\cR_E = \{PU_k^{\dagger}\}$ is lower bounded as follows 
    \begin{align}\label{eq:fid_min_bound}
        F_{\rm min}(\cR_{E}\circ \cA) & \geq \underset{|\psi\rangle \in \cC}{\rm min} \,\,\frac{1}{d^2}\sum\limits_{\m,k} |\langle \m |E_k^{\dagger}E_k |\m\rangle|^2,
    \end{align}
\end{lemma}
\begin{proof}

To begin with, we first note the following equality 
\begin{align}
    PE_k^{\dagger}E_kP &= (PA_k^{\dagger}- PA_k^{\dagger}W_{k-1} )(A_kP- W_{k-1}A_kP )\\
 \label{eq:PE_kE_KP=PA_kA_KP}   & = PA_k^{\dagger}A_kP - PA_k^{\dagger}W_{k-1}A_kP= PE_k^{\dagger}A_kP.
\end{align}

   The worst-case fidelity $F^2_{\rm min}(\cR_E \circ \cA)$ is the following
    \begin{align}
        F^2_{\rm min}(\cR_E \circ \cA) & = \underset{|\psi\rangle \in \cC}{\rm min}  \sum\limits_{k,l}^{N_E,N_A}|\Ll\psi| R_k A_lP|\psi\Rl|^2\\
   \label{eq:2nd_line}      F^2_{\rm min}(\cR_E \circ \cA)&\geq  \underset{|\psi\rangle \in \cC}{\rm min} \frac{1}{d^2} \sum\limits_{k}^{N_E}|\Ll\psi| R_k A_lP|\psi\Rl|^2\\
         &= \underset{|\psi\rangle \in \cC}{\rm min}  \frac{1}{d^2} \sum\limits_{k}^{N_E}|\Ll\psi| PU_k A_lP|\psi\Rl|^2\\
           & = \underset{|\psi\rangle \in \cC}{\rm min}  \frac{1}{d^2} \sum\limits_{k}^{N_E}|\Ll\psi| \Tilde{M}_{kk}^{-1/2}P E_k^{\dagger}A_kP |\psi\Rl|^2\\
\label{eq:5th_line}      &= \underset{|\psi\rangle \in \cC}{\rm min}  \frac{1}{d^2} \sum\limits_{k}^{N_E}|\Ll \psi| \Tilde{M}_{kk}^{-1/2}P E_k^{\dagger}E_kP |\psi\rangle|^2\\
      & = \underset{|\psi\rangle \in \cC}{\rm min} \frac{1}{d^2} \sum\limits_{k}^{N_E}|\Ll\psi| \Tilde{M}_{kk}^{1/2}|\psi\Rl|^2
    \end{align}
 Here $W_{k-1} =\sum\limits_{i=1}^{k-1}U_kP_KU_k^{\dagger}$. Here also the $\tr[.]$ is taken over the support of code-space.  We use the triangle inequality to arrive at the equality in Eq.\eqref{eq:2nd_line} and for the inequality in Eq.\eqref{eq:5th_line} we use Eq.\eqref{eq:PE_kE_KP=PA_kA_KP}.
  
\end{proof}

Now we note that to correct the errors $A_ks$ we first have to detect the orthogonal subspaces forms by the operators $E_kP$s. Total number of $E_KP$s,i.e, $N_E$ should be such that $F_{\rm ent}(\cR_E \circ \cA) > 1-\cO(\epsilon^{t+1}) $. Following the \cite{leung} we see that the total detection probability $P_{\rm det} = \sum_k\langle \psi|E_k^{\dagger}E_k|\psi\rangle$ should be greater than $F$ for a code to be $\cO(\epsilon^{t})$ correctable.

Since, after the orthogonalisation we are in the same framework provided by Leung et al. \cite{leung}, we argue that the $P_{\rm det}> F$ only if the difference between the  minimum and the maximum eigenvalue of $\tilde{M}_{kk}$ is $\cO(\epsilon^{t+1})$ then only the $F > 1- \cO(\epsilon^{t+1})$.

Now we can see that if the vectors $E_k|\m\rangle$  deviates as $\cO(\epsilon^{t+1})$ from the vectors $A_k|\m\rangle$ then the recovery $\cR_E$ corrects the noise upto $\cO(\epsilon)$.





\section{Syndrome detection circuit for the polar decomposition based recovery}\label{app:uni_syn}
\begin{figure}[t!]
    \centering
    \includegraphics[width=1\columnwidth]{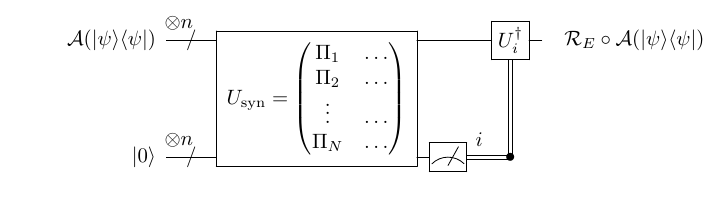}
    \caption{Circuit for the recovery operation $\cR_E$ in Eq.\eqref{eq:rec_new_1}. The $U_{\rm syn}$ is the unitary operation used for the syndrome measurement circuit. Upon measuring the ancillary qubit, which is initialized in $|0\rangle$, we obtained the syndromes, then depending on the measured syndrome outcome $i$, we apply the $U^{\dagger}_i$.   }
    \label{fig:univ_syn}
\end{figure}
So far, we have proposed a universal recovery based on the polar decomposition of the AQEC code. The polar decomposition-based recovery $\cR_E$ requires measuring the projectors $\Pi_k$, which are the projectors onto the orthogonal syndrome subspaces. One way of implementing the recovery $\cR_E$ is to decompose it into a two-outcome POVM \cite{lloyd2001engineering}. Then, one can follow either the binary tree implementation with a single ancilla \cite{shen2017} or the block-encoding and sequential techniques with two ancillary qubitd \cite{biswas2024noise}. However, here we propose a different circuit for the syndrome-based recovery with the AQEC code, which includes explicit syndrome measurements. 
The following Lemma proves the circuit in Fig. \ref{fig:univ_syn} implements the recovery $\cR_E$ in Eq.\eqref{eq:rec_new_1}. 

\begin{lemma}
   Any recovery channel that can be decomposed as a projective measurement projecting onto mutually orthogonal subspaces, followed by a unitary operation as in Eq.\eqref{eq:rec_new_1}, has a circuit as depicted in Fig. \ref{fig:univ_syn}. 
\end{lemma}
\begin{proof}
   We first note that to implement the recovery $\cR_E$, one needs to measure the syndromes by measuring the operators $P_k$. The operators $P_k$ satisfy the following inequality:
\begin{align}
\sum_k \Pi_k \leq I & \quad, \Pi_k = \Pi_k.\Pi_k.
\end{align}
We also note that the $\Pi_k$s are the projectors onto the orthogonal syndrome subspaces. To implement each $\Pi_k$, we can follow the block-encoding of $\Pi_k$, then implement $R_k$, measure the ancilla, and finally implement $U_k^{\dagger}$ if the ancilla is in $|0\rangle$.

To implement $\cR_E$, we can first design the isometric extension for each $\Pi_k$ as
\begin{align}
U_{\Pi_k}& = \begin{pmatrix}
P_k & * \\
\sqrt{I - \Pi_k} & *
\end{pmatrix},
\end{align}
and moreover follow the POVM-based circuit construction proposed in \cite{biswas2024noise}. The former block encoding-based procedure is probabilistic, while the POVM-based procedure is deterministic but approximate. Here, we propose a new method to implement $\cR_E$ through the circuit in Fig.\ref{fig:univ_syn}.

Firstly, we note that the operators $\Pi_k$, the projectors onto subspaces generated by the Kraus operators ${E_kP}$, form a CP but trace non-increasing map $\cP$. Therefore, for an [[n,k]] code, the $\Pi_k$s are of dimension $2^n \times 2^n$, and we can construct an isometric extension of the channel $\cP$ as follows:
\begin{align}
U_{\rm iso} &= \begin{bmatrix}
\Pi_1 \\
\Pi_2 \\
\vdots \\
\Pi_N
\end{bmatrix}_{N 2^n \times 2^n}.
\end{align}

To implement this isometric extension, we consider the following joint unitary operation between the system (the $n$ encoded data qubits) and the $n$-qubit ancilla, which is initialized in $|0\rangle^{\otimes n}$:
\begin{align}
U_{\rm syn} = \begin{pmatrix}
\Pi_1 & \hdots \\
\Pi_2 & \hdots \\
\vdots & \hdots\\
\Pi_N & \hdots
\end{pmatrix}_{2^n N \times 2^n N}.
\end{align}

To implement $U_{\rm syn}$, one can follow the algorithm outlined in Ref.\cite{biswas2024noise}. After implementing $U_{\rm syn}$, we measure the ancilla qubits, which are initialized in the $|0\rangle^{\otimes n}$ state, in the $Z$ basis (the ${0,1}$ computational basis). After the measurement, if the ancilla collapses to the $|k\rangle$-th state in the computational basis, the encoded data block is projected onto the following state:
\begin{align}
\rho_k & = \frac{\Pi_k \cA(|\psi\rangle\langle\psi|)\Pi_k}{\tr [\Pi_k\cA(|\psi\rangle\langle\psi|)]}.
\end{align}

Thus, the measurement of the ancilla projects the noisy state $\cA(|\psi\rangle\langle\psi|)$ onto the $k$-th syndrome subspace. The quantity $\tr [\Pi_k\cA(|\psi\rangle\langle\psi|)]$ is the probability of detecting the $k$-th syndrome. After obtaining the $k$-th outcome, we apply $U_k^{\dagger}$. Therefore, the final state is:
\begin{align}
\Tilde{\rho}_k &= U_k^{\dagger} \rho_k U_k.
\end{align}

Summing over all possible measurement outcomes on the ancilla qubits, the recovered state is:
\begin{align}
\cR_E\circ \cA(|\psi\rangle\langle\psi|)= \sum U_k^{\dagger} \Pi_k \cA(|\psi\rangle\langle\psi|) \Pi_k U_k.
\end{align}
\end{proof}
For stabilizer codes, $U_{\rm syn}$ is always a Clifford operation and can be simulated efficiently (in polynomial time). However, for non-stabilizer codes, the $U_{\rm syn}$ circuit may contain some non-Clifford gates.

\section{Petz vs Syndrome-Based Petz}\label{app:petz_vs_syn_petz}

In the main text, we showed that the syndrome-based Petz map satisfies a bound analogous to that of the Petz map, similar to how the optimal--optimal recovery bound is obtained. However, this bound alone is not sufficient to claim that the syndrome-based Petz recovery map is an optimal recovery for all codes and all noise models. Moreover, we observed that for the six-qubit code under amplitude-damping noise, the syndrome-based Petz map performs worse than the Petz map, whereas for the four-qubit codes it performs better.

In this section, we therefore examine the constraints on the codes and noise under which the syndrome-based Petz map can outperform the standard Petz map. We begin with the following lemma.

\begin{lemma}\label{lem:PE_kAlP=0}
    $PE_k^{\dagger}A_lP = 0$ for all $k>l$.
\end{lemma}

\begin{proof}
    First note that $Q_k=I - \sum_{i=1}^{k-1} U_i P_i U_i^{\dagger}$ is a projector, i.e., $Q_k = Q_k Q_k$, and therefore $E_k = Q_k A_k = Q_k Q_k A_k$.

    From the expression of $Q_k$, we have
    \begin{align}
        Q_k A_l = (I - \sum_{i=1}^{k-1} U_i P_i U_i^{\dagger}) A_l .
    \end{align}
    Since $l < k$, this becomes
    \begin{align}
        Q_k A_l &= A_l - \sum\limits_{i=1}^{l-1} U_i P_i U_i^{\dagger} A_l - \sum\limits_{i=l}^{k-1} U_i P_i U_i^{\dagger} A_l \\
        &= E_l - \sum\limits_{i=l}^{k-1} U_i P_i U_i^{\dagger} A_l .
    \end{align}
    Thus,
    \begin{align}\label{eq:proof_1}
        PE_k^{\dagger}A_l P
        = P E_k^{\dagger} E_l P
        - \sum\limits_{i=l}^{k-1} P E_k^{\dagger} U_i P_i U_i^{\dagger} A_l P .
    \end{align}
    The first term in Eq.~\eqref{eq:proof_1} vanishes for $k \neq l$. Now considering the algebraic form of $E_kP = Q_kA_k$, we have the second term of \eqref{eq:proof_1} as :
    \begin{align}
         \sum\limits_{i=l}^{k-1} PE_k^{\dagger} U_i P_i U_i^{\dagger} A_l P
         &= P\Big( \sum\limits_{i=l}^{k-1} A_k^{\dagger} U_i P_i U_i^{\dagger} A_l
        \nonumber \\
         & \quad- \sum\limits_{j=l,i=1}^{k-1,k-1} A_k^{\dagger} U_j P_j U_j^{\dagger} U_i P_i U_i^{\dagger} A_l \Big)P \\
         &= 0 .
    \end{align}
    Here we have used $P_jU_j^{\dagger}U_iP_i=\delta_{ij}P_i$
    Hence, for all $k>l$, $PE_k^{\dagger}A_lP = 0$.
\end{proof}

Next, to derive the entanglement fidelity under the syndrome-based Petz map, we introduce the following lemma.

\begin{lemma}\label{lem:PEKAlP=PEKA_l'P}
  If the support of $E_{k}P$ shares the similar vectors with the $A_{l}P$ and $A_{l'}P$ then 
  \begin{align}
      P_kE_{k}^{\dagger}A_lP = P_kE_{k}^{\dagger}A_{l'}P. 
  \end{align}
\end{lemma}

\begin{proof}
We first note that the projectors on to the support of $E_kP$ is the following 
\begin{align}
    U_kP_kU_{k}^{\dagger}= \Pi_{k}.
\end{align}
Now if the support of $A_lP$ and $A_{l'}P$ shares a same vector $|\psi\rangle$ with the support of the $E_kP$, then 
\begin{align}
    \Pi_kA_lP = \Pi_kA_{l'}P.
\end{align}
We note that $U_k= E_kP (PE_k^{\dagger}E_kP)^{-1/2}$, where the inverse is taken over the support of $ PE_k^{\dagger}E_kP$  . Therefore unisng this form of $U_k$ in the expression of $\Pi_k$ we have $ P_kE_{k}^{\dagger}A_lP = P_kE_{k}^{\dagger}A_{l'}P$
\end{proof}

We now have the ingredients required to compute the entanglement fidelity under the syndrome-based Petz map.

\begin{lemma}\label{lem:petz_ent_fid}
    The entanglement fidelity under the syndrome-based Petz map is
    \begin{align}
        F_{\rm Ent}& =\frac{1}{d^2}\bigg(
      \sum\limits_{k}^{N_E}
      \Big|\sum\limits_{\m} (\sqrt{\Tilde{M}})_{[\m k],[\m k]} \Big|^2 \\
      & +
      \sum\limits_{k}^{N_E}\frac{(N_A-k)(N_A+k+1)}{2}
      \Big|\sum\limits_{\m} (\sqrt{\Tilde{M}})_{[\m k],[\m k]} \Big|^2
      \bigg).
    \end{align}
\end{lemma}

\begin{proof}
    Using the Kraus representation of the syndrome-based Petz, the Lemma \ref{lem:PE_kAlP=0}-\ref{lem:PEKAlP=PEKA_l'P} and the Eq.\eqref{eq:PE_kE_KP=PA_kA_KP}, we obtain
 \begin{align}
      F_{\rm Ent}
      &= \frac{1}{d^2}\sum\limits_{k,l}^{N_E,N_A}
         \Big|\sum\limits_{\m,\n}
         \Tilde{M}^{-1/2}_{[\m k],[\n k]}
         \langle \n | E_k^{\dagger} A_l | \m \rangle \Big|^2 \\
      &= \frac{1}{d^2}\bigg(
      \sum\limits_{k}^{N_E}
      \Big|\sum\limits_{\m,\n}
      \Tilde{M}^{-1/2}_{[\m k],[\n k]} \Tilde{M}_{[\n k],[\m k]}
      \Big|^2 \nonumber \\
      &\qquad
      + \sum\limits_{k<l}
      \Big|\sum\limits_{\m,\n}
      \Tilde{M}^{-1/2}_{[\m k],[\n k]}
      \langle\n |E_k^{\dagger}A_l|\m\rangle \Big|^2
      \bigg) \\
      &= \frac{1}{d^2}\bigg(
      \sum\limits_{k}^{N_E}
      \Big|\sum\limits_{\m} (\sqrt{\Tilde{M}})_{[\m k],[\m k]} \Big|^2\nonumber
      \\
      & +
      \sum\limits_{k}^{N_E}\frac{(N_A-k)(N_A+k+1)}{2}
      \Big|\sum\limits_{\m} (\sqrt{\Tilde{M}})_{[\m k],[\m k]} \Big|^2
      \bigg).
 \end{align}
\end{proof}

To compare the syndrome-based Petz map with the Petz map, we note that the term  
\[
\frac{1}{d^2} \sum\limits_{k}^{N_E} \Big|\sum\limits_{\m} (\sqrt{\Tilde{M}})_{[\m k],[\m k]}\Big|^2
\]
is always smaller than  
\[
\frac{1}{d^2} \sum\limits_{k}^{N_E} \Big|\sum\limits_{\m} (\sqrt{M})_{[\m k],[\m k]}\Big|^2.
\]
Hence, the last term in $F_{\rm Ent}$ determines when the syndrome-based Petz map can outperform the Petz map.

For codes and noise models where the QEC matrix $M$ is block diagonal---i.e., different error subspaces are orthogonal, as in the four-qubit code---the modified QEC matrix $\Tilde{M}_{kk}$ is similar to $M_{kk}$. Suppose there are $N_s$ such error operators. We also note that $N_E < N_A$. Therefore, the difference  
\[
\Delta_F = F_{\rm Ent}(\cR_{P,\cE}\circ \cA) - F_{\rm Ent}(\cR_{P,\cA}\circ \cA)
\]
takes the form
\begin{align}
  \Delta_F &= \frac{1}{d^2} \sum\limits_{k=N_s+1}^{N_E}
  \Big( |\sum\limits_{\m}(\sqrt{\Tilde{M}})_{[\m k],[\m k]}|^2
  - |\sum\limits_{\m}(\sqrt{M})_{[\m k],[\m k]}|^2 \Big) \nonumber \\
    &\quad + \frac{1}{d^2} \sum\limits_{k}^{N_E}
    \frac{(N_A-k)(N_A+k+1)}{2}
    |\sum\limits_{\m} (\sqrt{\Tilde{M}})_{[\m k],[\m k]}|^2 .
\end{align}

As previously noted, the term  
\[
\sum\limits_{k=N_s+1}^{N_E}
\Big( |\sum\limits_{\m}(\sqrt{\Tilde{M}})_{[\m k],[\m k]}|^2
     -|\sum\limits_{\m}(\sqrt{M})_{[\m k],[\m k]}|^2 \Big)
\]
is negative. Therefore, the syndrome-based Petz map outperforms the Petz map only when the orthogonalisation algorithm yields a sufficiently large positive contribution from the final term.

\end{document}